\documentclass[12pt]{article}
\usepackage{authblk}
\usepackage[margin=1in]{geometry} 

\usepackage{stackrel}
\usepackage{amsmath} 
\usepackage{amssymb}
\usepackage{color}
\usepackage{enumerate}
\usepackage{wrapfig}
\usepackage{graphicx}
\usepackage{framed}
\usepackage[square,sort,compress,comma,numbers]{natbib}

\usepackage[all]{xy}

\def\be{\begin{equation}}
\def\ee{\end{equation}}
\def\bea{\begin{eqnarray}}
\def\eea{\end{eqnarray}}
\def\ba{\begin{array}}
\def\ea{\end{array}}

\def\bx{{\mathbf {x} }}

\let\<\langle
\let \>\rangle

\newtheorem{theorem}{Theorem}

\newtheorem{lemma}[theorem]{Lemma}

\newtheorem{remark}[theorem]{Remark}

\numberwithin{theorem}{section}
\newenvironment{proof}[1][Proof]{\textbf{#1.} }{\ \rule{0.5em}{0.5em}}

\newcommand{\rem}[1]{}

\newcommand{\de}{\delta}

\newcommand{\bvarphi}{\boldsymbol{\varphi}}

\newcommand{\bu}{\boldsymbol{u}}

\newcommand{\bX}{\boldsymbol{X}}

\newcommand{\bZ}{\boldsymbol{Z}}
\newcommand{\bY}{\boldsymbol{Y}}

\newcommand{\pp}[2]{\frac{\partial #1}{\partial #2}}

\newcommand{\dede}[2]{\frac{\delta #1}{\delta #2}}

\usepackage{framed}

\newcommand{\todo}[1]{\vspace{5 mm}\par \noindent
\framebox{\begin{minipage}[c]{0.95 \textwidth}
\tt #1 \end{minipage}}\vspace{5 mm}\par}

\begin{document}

\title{Thermodynamically consistent variational theory of porous media with a breaking component}

\author[1]{Fran{\c{c}}ois Gay-Balmaz}
\author[2]{Vakhtang Putkaradze$^{\star}$} 
\affil[2]{Department of mathematical and statistical sciences, University of Alberta, 
Edmonton, AB, T6G 2G1 Canada; email: putkarad@ualberta.ca, ($\star$) Corresponging author
} 
\affil[1]{
LMD - Ecole Normale Sup\'erieure de Paris - CNRS, 75005, Paris; email: francois.gay-balmaz@lmd.ens.fr
} 

\maketitle

\abstract{
If a porous media is being damaged by excessive stress, the elastic matrix at every infinitesimal volume separates into a `solid' and a `broken' component. The `solid' part is the one that is capable of transferring stress, whereas the `broken' part is advecting passively and is not able to transfer the stress. In previous works, damage mechanics was addressed by introducing the \emph{damage parameter} affecting the elastic properties of the material.  In this work, we take a more microscopic point of view, by considering the \textit{transition} from the `solid' part, which can transfer mechanical stress, to the `broken' part, which consists of microscopic solid particles and does not transfer mechanical stress. 
Based on this approach, we develop a thermodynamically consistent dynamical theory for porous media including the transfer between the `broken' and `solid' components, by using a variational principle recently proposed in thermodynamics.
This setting allows us to derive an explicit formula for the breaking rate, \emph{i.e.}, the transition from the `solid' to the `broken' phase, dependent on the Gibbs' free energy of each phase. Using that expression, we derive a reduced  variational model for material breaking under one-dimensional deformations. We show that the material is destroyed in finite time, and that the number of `solid' strands vanishing at the  singularity follows a power law. We also discuss connections with existing experiments on material breaking and extensions to multi-phase porous media.
\medskip 

\noindent {\bf Keywords}: Variational methods, Damage models, Thermodynamics, Poromechanics
\\
{\bf Statements and Declarations:} Authors declare no competing interests. 
}

\tableofcontents

\section{Introduction}

Many materials in biological and engineering applications contain a porous matrix filled with fluid. For large deformations or sudden impacts, some of the elastic connections forming the matrix can break, so the matrix can no longer transmit stress through the broken strands. In biological materials, the matrix can also heal itself, thus increasing the local number of elastic connections. The aim of this paper is to develop a variational theory for  porous media when the elastic matrix can break up. In principle, our theory will also be applicable to the repair of the matrix; however, in this paper we focus  exclusively on the process of irreversible damage to the matrix and subsequent dynamics.

The framework of Continuum Damage Mechanics (CDM) has been developed to address damage in solids due to large deformations and stresses \cite{kachanov1986introduction,lemaitre2012course}. There are many different types of damage in continuus materials, such as breaking bonds in polymer media, ungluing of cellulose fibers in wood damage, appearance of microcracks and voids in crystalline media and many others. CDM is a very well established field and has influenced the corresponding work in porous media. One of particular approaches in CDM is the introduction of the crack density tensor \cite{kachanov1993elastic} for anisotropic materials which describes the volumetric evolution of micro-cracks in materials, see \cite{homand1998continuum} for application of this work to brittle rocks. Variational approaches to the evolution of media with a continuum damage variable were also studied before in the context of solid mechanics, see in particular \cite{placidi2008variational,placidi2018two,placidi2018strain,placidi2020variational,abali2021novel}. The ideas coming from CDM were later used in applications to poroelastic media. In particular, \cite{hamiel2004coupled,lyakhovsky2007damage} use thermodynamic approach to derive the physically relevant expressions for the evolution of damage parameter and porosity as applied to geological materials. CDM has also been used for biomedical applications, such as corrosion of stents \cite{gastaldi2011continuum} and damage accumulation in bone cement \cite{jeffers2005damage}.
It is also believed that in many practical applications, the damage exhibited in the form of microscopic cracs and voids is manifested at a different spacial scale compared to that of the scales of pores in the porous media, thus necessitating nonlocality in equations \cite{bavzant2002nonlocal}. 
The model of evolution of damage parameter for poroelasticity was suggested in \cite{mobasher2017non} based on quasi-static extension of Biot's original equations of consolidation \cite{biot1941general}. The damage parameter, a dimensionless number between $0$ and $1$, is contributing to the degradation of elastic coefficients in the stress-strain relationships,  modeling the microscopic cracks and voids. 

While it is known that the damage to rock-like materials is done by creating microcracks and voids, the situation with fluid-filled porous biological materials is likely to be different. As the polymers forming the basis of the porous matrix are ripped apart by local forces, they create damaged parts that are no longer able to carry stress and are suspended in the surrounding fluid. The broken microparticles of polymers can be advected through the pores if they are small enough. Thus, it is not enough to consider the local value of porosity only as it does not specify the actual local water content. In reality, one must consider the three components of matter: `solid' which is able to carry the elastic stress, `broken' which consists of particles arising from the broken strands in the matrix, and `fluid'. 

Thus, at each infinitesimal volume larger than the local microstructure size, we consider three states of matter: the fluid denoted with the label $f$ , the solid denoted with the label $s$, and the `broken solid' denoted with the label $b$. The broken solid strands  behave similar to fluid: as isolated strands of material, they simply move in the fluid as small solid particles. The solid component can become broken under application of stress; the broken component can become solid under healing, a process that we will treat in the general theory but not consider in particular applications. The goal of this paper is to make a thermodynamically consistent variational theory of the breaking of the elastic matrix and subsequent evolution of the three components. The broken component will play the role of damage parameter in our theory, and the change of the local elastiticty will follow naturally from the existence of the broken and solid components in the Lagrangian, where only the solid component can carry the stress. Note that the ideas of damage as a phase transition from the broken to solid component was recently formulated in \cite{bucci2023damage} for non-porous materials in the quasi-static approximation. Our work extends these ideas by considering a fully time-dependent variational theory. In particular, our approach includes the solid phase as the porous matrix, takes into account the dynamics of the fluid in the variational principle, and incorporates thermodynamic effects to derive the consistent expressions for transitions between solid and broken components.

Our work will be based on the variational approach to the thermodynamics porous media without damage \cite{gay2022variational}, which, in turn, is based on the variational derivation of porous media motion taking into account purely mechanical effects \cite{FaFGBPu2020,FaFGBPu2020_2}. It is unrealistic to provide a detailed description of the different models of porous media here, thus, we refer the reader to the reviews \cite{de2000contemporary,de2005trends} and in particular to the fundamental treatises \cite{coussy1995mechanics,de2012theory} for historical introduction and the background of different approaches to porous media modeling. Because of the complexity of fluid-structure interactions involved in porous media dynamics, variational methods were particularly useful for deriving the theories of mechanical motion of porous media  \cite{bedford1979variational,aulisa2007variational,aulisa2010geometric,lopatnikov2004macroscopic,lopatnikov2010poroelasticity}. Of particular interest to this paper are the works on the Variational Macroscopic Theory of Porous Media (VMTPM) which was formulated in its present form in 
\cite{dell2000variational,placidi2008variational,sciarra2008variational,madeo2008variational,dell2009boundary,serpieri2011formulation,serpieri2015variationally,serpieri2016general,auffray2015analytical,serpieri2016variational,travascio2017analysis}, also summarized in a recent book \cite{serpieri2017variational}. Further progress in consistent variational approach to the mechanics of porous media was achieved by the authors recently 
\cite{FaFGBPu2020}, in terms of the fluid content being a constraint in the fluid's incompressibility, and the fluid pressure being the Lagrange multiplier related to the incompressibility.  That description in \cite{FaFGBPu2020} was based on the classical Arnold description of incompressible fluid as geodesic motion \cite{arnold1966geometrie}. The thermodynamics effects were included in the variational principle in \cite{gay2022variational}, which allowed for the conservation of total energy and the use of second law of thermodynamics to derive thermodynamically consistent laws of motion generalizing the Darcy-Brinkman model of porous media \cite{brinkman1949calculation,brinkman1949permeability,kannan2008flow,srinivasan2014thermodynamic}. The present paper builds on these considerations  to model the breaking of the porous media using a variational approach incorporating both mechanical and thermodynamics effects.

The paper is organized as follows. First, in Section \ref{sec:setup}, we introduce the needed variables and physically relevant Lagrangian functions. In Section \ref{sec:review_heat_conduction} we provide the background on the variational principles applied to irreversible processes, in both the Lagrangian (material) and the  Eulerian (spatial) descriptions. In Section \ref{sec:Thermo_comp}, we derive the equations of motion for a porous media with breaking component, containing the momenta and density equations for each component (solid, broken, fluid), as well as the entropy evolution. Based on general considerations such as the non-decrease of entropy, we arrive at the thermodynamically consistent forms for the friction forces and stresses, as well as for the rate of breaking of the elastic matrix, \emph{i.e.}, the rate of transition from solid to broken components. In Section \ref{sec:1D_motion}, we reduce the system to one-dimensional spatial motion and derive some exact reductions to ODEs. These reductions yield predictions for matrix break-up that are important for experiments. The reduced system is shown to arise from a discrete variational approach in Section \ref{sec:1D_rip}, which provides a general modeling tool for the derivation of reduced models that are consistent with the two laws of thermodynamics.

\section{Setup of the problem: variables and Lagrangians}\label{sec:setup}

In this Section we describe the variables needed for the description of a porous media with fluid and solid (unbroken and broken) components. We also give the conservation laws for each component as well as the total Lagrangian of the system.

\subsection{Definition of variables for porous media} 

\subsubsection{Observed and actual densities}

For the description of porous media, it is important to make a distinction between the observed density of the fluid or the solid in a given volume, and the actual density of fluid filing the pores or elastic material comprising the matrix. The observed, or Eulerian, density of fluid is defined as the coefficient of proportionality between the mass of fluid contained in the given Eulerian volume $\mbox{d}^3 \boldsymbol{x}$, centered at the spatial point $ \boldsymbol{x}$, and the mass contained in that volume, and similarly for the broken and unbroken solid: 
\begin{equation} 
\begin{aligned} 
\mbox{d} m_f (t, \boldsymbol{x}) & = \rho_f(t, \boldsymbol{x}) \mbox{d}^3\boldsymbol{x}
\\
\mbox{d} m_b (t, \boldsymbol{x}) & = \rho_b(t. \boldsymbol{x}) \mbox{d} ^3\boldsymbol{x}
\\
\mbox{d} m_s (t, \boldsymbol{x}) & = \rho_s(t, \boldsymbol{x}) \mbox{d} ^3\boldsymbol{x} \, . 
\end{aligned} 
\label{solid_def} 
\end{equation}
The actual density of the fluid and solid is the density of the material filling the pores (for example, density of gas in the pores for fluid) or, correspondingly, density of the elastic material comprising the matrix (\emph{e.g.}, rubber). The actual densities will be denoted with bars $\bar \rho_f$, $\bar \rho_s$, $\bar \rho_b$. 

If $\phi_s(t, \boldsymbol{x})$ is the volume fraction of the fluid, and  $\phi_b(t,\boldsymbol{x})$ is the volume fraction of the broken solid, and we assume that the fluid fills the pores completely, the actual and Eulerian densities are related by 
\begin{equation} 
\rho_s = \phi_s \bar \rho_s \, , \quad 
\rho_b = \phi_b \bar \rho_b \, , \quad 
\rho_f = (1-\phi_s-\phi_b) \bar \rho_f \, .
\label{bar_no_bar} 
\end{equation}

\subsubsection{Configuration of the fluid and elastic components}\label{subsec_config}

Suppose $ \mathcal{B} _s$, $\mathcal{B}_b$, and $ \mathcal{B} _f$ denote the reference configurations containing the elastic and fluid labels $\bX$, $\bY$, and $\bZ$. The motion of the unbroken and broken elastic bodies (indexed by $s$ and $b$) and the fluid (indexed by $f$) is defined by three time dependent maps: 
\begin{enumerate} 
\item The Lagrangian mapping of particles of the solid part: $ \boldsymbol{\varphi }_s (t,\_\,): \mathcal{B} _s \rightarrow \mathbb{R} ^3  $,
\item The Lagrangian mapping of particles of the broken part: $ \boldsymbol{\varphi }_b (t,\_\,): \mathcal{B} _b \rightarrow \mathbb{R} ^3  $, and 
\item The Lagrangian mapping of particles of the fluid part: $\boldsymbol{\varphi}_f (t,\_\,): \mathcal{B} _f \rightarrow \mathbb{R} ^3$\,.
\end{enumerate} 
The spatial variables are then defined as 
\begin{equation} 
\boldsymbol{x} = \boldsymbol{\varphi }_s  (t,\bX),\quad \boldsymbol{x} = \boldsymbol{\varphi }_b  (t,\bY)\,,  \quad\text{and}\quad \boldsymbol{x} = \boldsymbol{\varphi }_f (t,\bZ)\,.
\label{Spatial_x_def}
\end{equation} 
We assume that there is no fusion of either fluid or elastic body particles, so the map $ \boldsymbol{\varphi }_s$, $ \boldsymbol{\varphi }_b$  and $ \boldsymbol{\varphi}_f $ are embeddings for all times $t$. 
We also assume that the fluid cannot escape the porous medium
or create voids, so at all times $t$, the domains occupied in space by the fluid, the broken and unbroken elastic bodies ${\mathcal B}_{t,k}= \boldsymbol{\varphi }_k (t,\mathcal{B}_k)$, $k=f,s,b$,
coincide: 
${\mathcal B}_{t,k}={\mathcal B}_{t}$. Finally, we shall assume for simplicity that the domain ${\mathcal B}_t$ does not change with time, and will simply call it ${\mathcal B}$, hence $ \boldsymbol{\varphi} _k (t,\_\,): \mathcal{B} _k \rightarrow \mathcal{B} $, $k=f,s,b$ are diffeomorphisms for all time $t$.  Without loss of generality, we can set $ \mathcal{B} _k=\mathcal{B} $, so that $ \boldsymbol{\varphi} _k  (t,\_\,)\in \operatorname{Diff}( \mathcal{B} )$ are diffeomorphisms of $ \mathcal{B} $ for $k=f,s,b$.

We shall note that the framework of our theory can also be adapted to the cases when the domains $\mathcal{B}_{s,b,f}$ do not coincide.  For example, in the process of hydraulic fracturing, the voids are created by breaking up the elastic matrix in a narrow area, so the broken component will be concentrated mostly in that area before eventually being washed off with fluid flow. In this paper, we will put $\mathcal{B}_{s,b,f}=\mathcal{B}$ for simplicity.
Alternatively, one can separate the whole domain into sub-domains connected by a moving boundary, separating the fluid and the broken component from the mostly unbroken component. These two approaches are similar to methods used in simulation of free-surface flows. In this paper we are going to consider the case when all three components are present in all the parts of the volume, which is similar in its approach to the volume-of-fluid method \cite{weymouth2010conservative}. This method is more easily treatable analytically and can allow direct comparison with the previous literature in the field. The method of separating volume into sub-volumes can also be used, which we will pursue in a follow-up work.
The extension to the case of the fluid escaping the boundary is possible, although it requires appropriate modifications in the variational principle that will also be considered in a future work.

\subsubsection{Velocities of the fluid and elastic components in the spatial frame}

The velocities $\bu_k$, measured relative to the fixed coordinate system, \emph{i.e.}, in the Eulerian representation, are given by 
\begin{equation} 
\bu_k(t,\boldsymbol{x} )=\partial_t  \boldsymbol{\varphi} _k  \big(t, \bvarphi_k^{-1}(t,\boldsymbol{x})\big)\, ,k=f,s,b,
\label{vel_def} 
\end{equation} 
for all $\boldsymbol{x} \in \mathcal{B}$.  Note that since the diffeomorphisms $ \boldsymbol{\varphi} _k$, $k=f,b,s$ have to preserve the material points  on the boundaries, the vector fields $\bu_k$ are tangent to the boundary, \emph{i.e.},
\begin{equation} 
\label{free_slip} 
\bu_k\cdot \boldsymbol{n} =0\, ,
\end{equation}
where $ \boldsymbol{n}$ is the unit normal vector field to the boundary. 
One can alternatively impose that $\bvarphi_k$, for some $k$, are prescribed on the boundary. In this case, one gets no-slip boundary conditions for those $k$: 
\begin{equation} \label{no_slip} 
\bu_k|_{ \partial \mathcal{B} }=0\,.
\end{equation}

\subsubsection{Continuity equations for the mass densities}\label{conservation_mass}

Let us look at the continuity equations for $\rho_k(t,\boldsymbol{x} ) \mbox{d}^3 \boldsymbol{x}$, where $\rho_k$ is the observed (Eulerian) density of component $k=f,s,b$.

In the most general and detailed form, the continuity equations can be written as follows
\begin{equation} 
\partial_t \rho_k + \operatorname{div} (\rho_k \bu_k ) = \sum_{\ell} \left( q_{k\ell}^+ - q_{k\ell}^- \right) \,, \qquad  k=f,s,b\,,
\label{dens_evol_gen}
\end{equation} 
where 
$q_{k\ell}^{+}\geq 0$ is the transfer rate of component $\ell$ being added to the component $k$, and $q_{k\ell}^{-}\geq 0$ is the transfer rate of component $\ell$ being added to the component $k$, which are functions of state variables of the system. The total transfer rate from $\ell$ to $k$ component is given by $q_{k\ell}:= q_{k\ell}^+- q_{k\ell}^-$. This distinction is important for multi-component transfers.  The conservation of total mass  follows from $q_{k\ell}^+=q_{\ell k}^-$ which implies $q_{k\ell} = - q _{\ell k}$.

In this paper, we will only consider the transfer between broken and solid, which simplifies the considerations quite a bit.


In the case of material break-up,
the threads are broken when excessive deformation or stress is applied to them. In biological materials, there may also be a process of solid matrix repair, or coagulation of the broken component to become solid again. Thus, there is no change in the amount of fluid particles; there is a rate of exchange $q_b:=q_{bs}^+=q_{sb}^-$ of materials from solid to broken components due to breaking of the solid matrix; and there is a rate of exchange from broken to solid components $q_r=q_{sb}^+=q_{bs}^-$ due to the repair or coagulation process. Thus, the exchange rates between components in our formulation are given by
\begin{equation} 
q_{fs}=q_{fb}=0, \quad q_{sb}=q_r-q_b, 
\label{coagulation_setup}
\end{equation} 
where $q_r$ can be given by biological or chemical conditions on the restoring action of the strands, and the breaking rate $q_b$ can be a function of either the elastic stress $ \boldsymbol{\sigma} _{\rm el}$ or Finger's deformation tensor $b$. 
As the coagulation or break-up processes will generate or consume heat, there will be a contribution to the thermodynamic equations and corresponding variations, as we show below.

In what follows, we shall focus on materials where only break-up is possible, while coaguation, or material repair, does not occur, i.e., $q_b\geq 0$ and $q_r=0$. The case of coagulation can be treated similarly to the theories derived in this paper, and will be pursued in a future work.

\subsection{The Lagrangian function}
\label{subsec_Lagr}

The Lagrangian of the porous medium is the sum of the kinetic energies of the fluid and elastic body (with solid and broken components) minus the potential energy for each material, giving the expression:
\begin{equation}\label{Lagr_def}
\begin{aligned} 
&\ell  (\bu_f,\bu_s, \bu_b,   \rho_f,  \rho_s,  \rho  _b, s_f, s_s,  s_b, b, \phi_s, \phi _b) \\
&= \int_{\mathcal{B}}
 \Big[\frac{1}{2}  \rho_f|\bu_f|^2 - \rho_f e_f(\bar \rho_f,s_f/ \rho  _f) + \frac{1}{2}\rho_b  |\bu_b|^2 -  \rho_b e_b(\bar \rho_b,s_b/ \rho  _b)\\
&\hspace{4.5cm}+ \frac{1}{2}\rho_s  |\bu_s|^2 -  \rho_s e_s(\bar \rho_s,s_s/ \rho  _s, b) \Big] \, {\rm d}^3 \boldsymbol{x}\,,
\end{aligned}
\end{equation}
with $ \rho  _f$, $ \rho  _s$, $ \rho  _b$ the observed mass densities, $ s_f$, $s_s$, $s_f$ the entropy densities, and $b$ the Finger deformation tensor of the solid component.
Here, $\bar \rho_f$, $\bar \rho_s$, $\bar\rho_b$ are the actual densities of the fluid inside the pores and of the material composing the elastic matrix as defined by \eqref{bar_no_bar} and $e_{f}(\bar \rho_{f},s_f/ \rho_f )$, $e_{s}(\bar \rho_{s},s_s/ \rho   _s,b )$, $e_{b}(\bar \rho_{b},s_b/ \rho_b)$ are the specific internal energies of each component. The expression \eqref{Lagr_def} explicitly separates the contribution from each component in simple physically understandable terms, similarly to our description in \cite{FaFGBPu2020}.

Note that in \eqref{Lagr_def} we used the \textit{observed} densities $\rho_k$ multiplying the values of the specific internal energies, and the \textit{actual} densities $\bar \rho_k$ in the functional expressions for the internal energies. The description in terms of actual densities $\bar \rho_{k}$ in the internal energies is more convenient for thermodynamic considerations, since $\bar \rho_k$ explicitly depend on entropy through the equations of state for a given material. We have also combined the thermal and elastic energy for the solid into a single internal energy function $e_s(\bar \rho_s,s_s/ \rho  _s, b)$ for convenience.
Even though $e_b$ identifies the energy of the broken solid component, it does not transmit the stress and thus cannot depend on the deformation tensor.

\section{Variational formulation for a viscous medium with heat conduction}
\label{sec:review_heat_conduction} 

In this Section, we review the variational  theory of a single-component fluid or elastic medium with heat conduction and viscosity, based on the variational thermodynamics approach developed in \cite{gay2017lagrangian,gay2017lagrangian2,GBYo2019}. This work was generalized for porous media in \cite{gay2022variational}, in the case where the elastic medium consists of a single component. We will use this variational thermodynamics framework in Section \ref{sec:Thermo_comp} to treat the case of a two-component ('solid' and 'broken') elastic media which includes the additional irreversible processes of transitions between the components, where both components are embedded in the fluid.

\subsection{Lagrangian (material) description}\label{var_therm_mat}

In the material description, the variational formulation of thermodynamics is an extension of the Hamilton principle of continuum mechanics. We recall the material description here for pedagogical reasons since the variational formulation takes its simplest form. The variational formulation in the spatial (Eulerian) description is then derived from it, see the next section, and turns out to be a useful tool for porous media modelling.

For a single-component fluid or elastic medium, we take the Lagrangian $L$ to be a function of the configuration diffeomorphism $\bvarphi(t, \mathbf{X} )$, see \S\ref{subsec_config}, the material velocity $ \boldsymbol{V}(t, \mathbf{X} )= \dot{\bvarphi}(t, \mathbf{X} )$, and the entropy density $S(t, \bX)$. It also depends parametrically on the density $\varrho_0(\bX)$ and on the Riemannian metric $G_0(\bX)$ on the reference configuration which can be chosen as the Euclidean one when $ \mathcal{B} \subset \mathbb{R} ^3$. The metric $G_0(\bX)$ is needed to introduce the Eulerian deformation tensors.

For each fixed $ \varrho _0( \mathbf{X} ) $ and $G_0( \mathbf{X} )$, we can write the Lagrangian as a function $L: T \operatorname{Diff}_0( \mathcal{B}) \times \operatorname{Den}( \mathcal{B} ) \rightarrow \mathbb{R}$, where $ \operatorname{Diff}_0( \mathcal{B} )\ni \boldsymbol{\varphi} $ is the group of diffeomorphisms of $ \mathcal{B} $ keeping $ \partial \mathcal{B} $ pointwise fixed, corresponding to no-slip boundary conditions, $T\operatorname{Diff}_0( \mathcal{B} )\ni (\boldsymbol{\varphi}, \dot{\boldsymbol{\varphi}})$ is its tangent bundle, and $\operatorname{Den}( \mathcal{B} )\ni S$ is the space of densities on $ \mathcal{B} $. In the material description, a standard expression is given by
\begin{equation}\label{General_L_continuum_thermo}
L( \boldsymbol{\varphi }  , \dot{ \boldsymbol{\varphi }  }, S)= \int_ \mathcal{B} \left[ \frac{1}{2} \varrho _0 | \dot{\boldsymbol{\varphi } }|-  \varrho_0  \,\mathcal{E} (  \nabla \boldsymbol{\varphi}  , \varrho_0 , S, G_0) \right] {\rm d}^3 \boldsymbol{X}\,,
\end{equation}
with $\mathcal{E}$ the internal energy in the material description, subject to the usual covariance assumptions, \cite{marsden1994mathematical,GBMaRa12}. In particular, due to the material covariance assumption, the internal energy can be written in terms of the Eulerian variables $ \rho  , s, b$ see \S\ref{subsec_Eulerian_descr}, as
\begin{equation}\label{E_to_e} \mathcal{E} ( \nabla \boldsymbol{\varphi} , \varrho _0, S, G_0)= e( \rho  , s, b) \circ \boldsymbol{\varphi},
\end{equation} 
with $e$ the specific internal energy in the Eulerian description.


The critical action principle for thermodynamics needs the introduction of two additional variables: $ \Sigma(t ,\boldsymbol{X}) $ which is identified with the entropy generated by the irreversible processes (unlike $S(t, \mathbf{X} )$, which is the total entropy), and $ \Gamma (t ,\boldsymbol{X}) $ identified with the thermal displacement \cite{gay2017lagrangian2}. With that notation, the critical action principle for a heat conducting viscous continuum reads \cite{gay2017lagrangian,gay2017lagrangian2}: 
\begin{equation}\label{VC_review} 
\delta \int_0^ T \Big[L (\boldsymbol{\varphi }  , \dot{ \boldsymbol{\varphi }  }, S) + \int_ \mathcal{B} (S- \Sigma ) \dot \Gamma \,{\rm d}^3 \boldsymbol{X} \Big]{\rm d}t=0\, ,
\end{equation} 
subject to the \textit{phenomenological constraint}
\begin{equation}\label{PC_review} 
\frac{\delta  L}{ \delta  S}\dot \Sigma = - \boldsymbol{P}
: \nabla \dot{\boldsymbol{\varphi }} + \boldsymbol{J}_S \cdot \nabla \dot{ \Gamma }
\end{equation} 
and with respect to variations $ \delta \boldsymbol{\varphi } $, $ \delta S$, $ \delta \Sigma $, $ \delta \Gamma $ subject to the \textit{variational constraint}
\begin{equation}\label{VConstr_review} 
\frac{\delta  L}{ \delta  S}\delta \Sigma = - \boldsymbol{P}
: \nabla \delta \boldsymbol{\varphi } + \boldsymbol{J}_S \cdot \nabla \delta \Gamma\,.
\end{equation}
The tensor $\boldsymbol{P}$ is the Piola-Kirchhoff viscous stress tensor and $ \boldsymbol{J}_S$ is the entropy flux density in Lagrangian representation, see \cite{gay2017lagrangian2}. They are the Lagrangian objects corresponding to the more widely used Eulerian viscous stress tensor $ \boldsymbol{\sigma }$ and Eulerian entropy flux density $ \boldsymbol{j}_s$, see below.

An application of \eqref{VC_review}--\eqref{VConstr_review} yields, in addition to the equations of motion in the Lagrangian frame that we do not present here, the conditions 
\begin{equation} \dot  \Sigma = \dot  S +  \operatorname{DIV} \boldsymbol{J}_S 
\;\; \mbox{ and } \;\;\dot  \Gamma = - \frac{\delta L}{\delta S} = T
\label{additional_constr} 
\end{equation} 
with $T$ being the temperature, see \cite{gay2017lagrangian2}, with the notation $\operatorname{DIV}$ indicating the divergence operator in the reference frame, as opposed to the divergence operator $\operatorname{div}$ in the spatial frame. These conditions attribute to $ \Sigma $ and $\Gamma $ their physical meaning mentioned above.

Since it is $\Sigma$ that describes the entropy of irreversible processes, from the second law of thermodynamics, we must have 
\begin{equation} 
\dot \Sigma \geq 0 \, ,
\end{equation} 
whereas $\dot S$ does not necessarily has to have a particular sign.

\begin{remark}[Structure of the variational formulation]\label{rmk_structure}\rm Observe that \eqref{VC_review}--\eqref{VConstr_review} is an extension of the Hamilton principle $ \delta \int_0^T L( \boldsymbol{\varphi} , \dot{\boldsymbol{\varphi}} ) {\rm d} t=0$ of continuum mechanics which involves two types of constraints: the constraint \eqref{PC_review} on the critical curve and the constraint \eqref{VConstr_review} on the variations to be used when computing this critical curve. One passes from \eqref{PC_review} to \eqref{VConstr_review} by formally replacing the time rate of changes by $\delta$-variations for each irreversible processes, such as $F_i \dot x^i \leadsto F_i \delta x^i$ in the case of an irreversible process due to a friction force for a finite dimensional system, see \cite{gay2017lagrangian,gay2017lagrangian2,GBYo2019}. This variational formulation is reminiscent from the Lagrange-d'Alembert principle used in nonholonomic mechanics. A finite dimensional version of the variational formulation \eqref{VC_review}--\eqref{VConstr_review}  will be used as a modeling tool in Section \ref{sec:1D_rip}.
\end{remark}

\subsection{Eulerian (spatial) description}\label{subsec_Eulerian_descr}

The Eulerian variables are the velocity $\bu(t, \boldsymbol{x})$, mass density $ \rho(t, \boldsymbol{x})$, entropy density $ s(t, \boldsymbol{x})$, and Finger deformation tensor $b(t, \boldsymbol{x})$ defined from the material variables $ \boldsymbol{\varphi} (t, \boldsymbol{X})$, $\dot{ \boldsymbol{\varphi}} (t, \boldsymbol{X})$, $\varrho _0 (t, \boldsymbol{X})$, $S (t, \boldsymbol{X})$, and $G_0 (t, \boldsymbol{X})$ in the usual way. In particular, the Finger deformation tensor is the push-forward of the inverted metric by the configuration diffeomorphism $ \boldsymbol{\varphi} $:
\begin{equation}\label{def_b} 
b= \boldsymbol{\varphi} _* G_0 ^{-1},
\end{equation} 
see \cite{marsden1994mathematical,GBMaRa12}.
In the case when $\mathcal{B}$ is a subset of $\mathbb{R}^3$ and $G_0$ is the identity tensor, the Finger deformation tensor $b$ is written in terms of the deformation gradient tensor $\mathbb{F}= \nabla \boldsymbol{\varphi} $ as
\begin{equation} 
b = \mathbb{F} \cdot \mathbb{F}^T
\, , \quad 
b^{ij}(t, \boldsymbol{x}) = \pp{x^i}{X^k}\pp{x^j}{X^k}( \bvarphi^{-1}(t, \boldsymbol{x}))\,.
\label{Finger_explicit}
\end{equation}
Note that from $ \boldsymbol{\varphi}  \in \operatorname{Diff}_0( \mathcal{B})$, we have the no-slip boundary conditions $ \boldsymbol{u}=0$ on $ \partial \mathcal{B} $.
We also consider the thermal displacement $ \gamma (t, \boldsymbol{x}) $, internal entropy density $ \sigma  (t, \boldsymbol{x}) $, Cauchy stress $ \boldsymbol{\sigma } (t, \boldsymbol{x}) $, and entropy flux $ \boldsymbol{j}_s (t, \boldsymbol{x}) $ defined as the Eulerian quantities associated to $ \Gamma  (t, \boldsymbol{X}) $, $ \Sigma  (t, \boldsymbol{X}) $, $ \boldsymbol{P} (t, \boldsymbol{X}) $, and $ \boldsymbol{J}_S (t, \boldsymbol{X}) $, see \cite{gay2017lagrangian2}.

With these definitions, the spatial version of \eqref{VC_review}--\eqref{VConstr_review} gives the variational formulation
\begin{equation}\label{VC_review_spat}
\delta \int_0^ T \Big[\ell (\boldsymbol{u}  , \rho  , s, b) + \int_ \mathcal{B} (s - \sigma ) D_t \gamma   \,{\rm d}^3 \boldsymbol{x} \Big]{\rm d}t=0\, ,
\end{equation} 
subject to the \textit{phenomenological constraint}
\begin{equation}\label{PC_review_spat}
\frac{\delta  \ell}{ \delta  s} \bar D_t \sigma = - \boldsymbol{\sigma } 
: \nabla  \boldsymbol{u}  + \boldsymbol{j}_s \cdot \nabla D_t  \gamma 
\end{equation} 
and with respect to variations
\begin{equation}\label{EP_variations} 
\delta \boldsymbol{u}= \partial _t \boldsymbol{\eta} + \boldsymbol{u} \cdot \nabla \boldsymbol{\eta} - \boldsymbol{\eta} \cdot \nabla \boldsymbol{u}, \quad \delta \rho  = - \operatorname{div}( \rho  \boldsymbol{\eta} ), \quad \delta b= - \pounds _ { \boldsymbol{\eta} }b,
\end{equation} 
$\delta s$, $\delta \sigma$, and $\delta \gamma$ 
subject to the \textit{variational constraint}
\begin{equation}\label{VConstr_review_spat} 
\frac{\delta \ell}{ \delta  s}\bar D_ \delta  \sigma = - \boldsymbol{ \sigma } 
: \nabla \boldsymbol{\eta}  + \boldsymbol{j}_s \cdot \nabla D_ \delta  \gamma\,.
\end{equation}

We recall that $ \boldsymbol{\eta}(t, \boldsymbol{x})= \delta \boldsymbol{\varphi} (t, \boldsymbol{\varphi} ^{-1} (t, \boldsymbol{x}))$ denotes the variation of the fluid trajectories in the Eulerian frame. In the expression of $ \delta b$, $\pounds _{ \boldsymbol{\eta}} b$ denotes the Lie derivative of the symmetric contravariant tensor $b$ in the direction $\boldsymbol{\eta}$, which follows from \eqref{def_b}. We have introduced the notations
\begin{equation}\label{D_t_notation}
\begin{aligned} 
&D_t f= \partial _t f + \boldsymbol{u} \cdot \nabla f &\qquad  & D_ \delta  f= \delta  f + \boldsymbol{\eta } \cdot \nabla f \\
&\bar D_t f = \partial _t f + \operatorname{div}(f \boldsymbol{u})  & \qquad  & \bar D_ \delta f = \delta  f + \operatorname{div}(f \boldsymbol{\eta })
\end{aligned}
\end{equation}
for the Lagrangian time derivative and Lagrangian variations of a scalar function and a density.

A direct application of the variational principle \eqref{VC_review_spat}--\eqref{VConstr_review_spat} yields the general equations of motion for a heat conducting viscous continuum with Lagrangian $\ell( \boldsymbol{u}, \rho  , s, b)$ in Eulerian coordinates as
\begin{equation}\label{reduced_EL_thermo_init} 
\left\{
\begin{array}{l}
\vspace{0.2cm}\displaystyle \partial_t \frac{\delta\ell}{\delta \boldsymbol{u} }+\pounds_{\boldsymbol{u}}\frac{\delta\ell}{\delta\boldsymbol{u}} =\rho\nabla \frac{\delta\ell}{\delta\rho}  +s\nabla \frac{\delta\ell}{\delta s}- \frac{\delta\ell}{\delta b}:\nabla b- 2\operatorname{div} \left( \frac{\delta\ell}{\delta b}\cdot b \right) + \operatorname{div} \boldsymbol{\sigma} \\
\displaystyle\frac{\delta \ell}{\delta s}( \bar D_t s + \operatorname{div} \boldsymbol{j}_s) = - \boldsymbol{\sigma} : \nabla \boldsymbol{u}- \boldsymbol{j}_s \cdot \nabla \frac{\delta \ell}{\delta s}
\end{array} \right. 
\end{equation} 
together with the conditions
\begin{equation}\label{two_extra_conditions_fluid} 
\bar D_t \sigma = \bar D_t s+ \operatorname{div} \boldsymbol{j}_s \qquad\text{and}\qquad \bar D_t \gamma = - \frac{\delta \ell}{\delta s}\,.
\end{equation} 
From these two conditions, $\bar D_t \sigma $ is interpreted as the total entropy generation rate density and $D_t \gamma $ is the temperature, hence $ \gamma$ is the thermal displacement.
The equations for $ \rho $ and $b$ are
\[
\partial _t \rho  + \operatorname{div}( \rho  \boldsymbol{u}  )=0 \quad\text{and}\quad  \partial _t b + \pounds _{ \boldsymbol{u}} b=0,
\]
which follow from the definition of $ \rho  $ and $b$ in terms of $ \varrho _0$ and $G_0$. Also, the variations $ \delta  \gamma $ at the boundary yields the insulated boundary conditions $ \boldsymbol{j}  _s \cdot \boldsymbol{n}=0$ on $ \partial \mathcal{B}$.
In the fluid momentum equations above, $ \pounds _ {\bu} \frac{\delta \ell}{\delta \bu}$ denotes the Lie derivative of the fluid momentum density $\frac{\delta \ell}{\delta \bu}$ in the direction $\bu$, explicitly given as $\pounds _ {\bu} \boldsymbol{m}  = \boldsymbol{u} \cdot \nabla \boldsymbol{m}  + \nabla \boldsymbol{u}^\mathsf{T} \boldsymbol{m} + \boldsymbol{m} \operatorname{div} \boldsymbol{u}$.

By using the standard expression of the Lagrangian 
\begin{equation} 
\label{Eulerian_Lagrangian}
\ell(\bu, \rho  , s, b) = \int_{\cal B} \Big[\frac{1}{2} \rho |\bu|^2 - \rho e(\rho,s/ \rho  ,b)\Big]{\rm d}^3 \boldsymbol{x}\, ,
\end{equation} 
which is the Eulerian form of \eqref{General_L_continuum_thermo}, see \cite{GBMaRa12}, the equations of motion \eqref{reduced_EL_thermo} give the following system of equations for a viscous and heat conducting continuum
\begin{equation}\label{reduced_EL_thermo} 
\left\{
\begin{array}{l}
\vspace{0.2cm}\displaystyle
\rho  (\partial _t \boldsymbol{u} + \boldsymbol{u} \cdot \nabla \boldsymbol{u}) = - \nabla p + \operatorname{div} \boldsymbol{\sigma}_{\rm el}+ \operatorname{div} \boldsymbol{\sigma}\\
\displaystyle T(\bar D_t s + \operatorname{div} \boldsymbol{j}_s) =  \boldsymbol{\sigma} : \nabla \boldsymbol{u}- \boldsymbol{j}_s \cdot \nabla T,
\end{array} \right. 
\end{equation} 
with $p= \rho  ^2 \frac{\partial e}{\partial \rho  } $ the pressure, $T= \frac{\partial e}{\partial \eta } $ the temperature, and $ \boldsymbol{\sigma} _{\rm el}=2 \rho \pp{e}{b}\cdot b$ the elastic stress.
We refer to \cite{gay2017lagrangian,gay2017lagrangian2,GBYo2019,gay2017variational} for the statement of the variational formulation in both the material and spatial description as well as the detailed computations and several applications and extensions.

\begin{remark}[Structure of the variational formulation]\label{rmk_structure_Eulerian}\rm The variational formulation \eqref{VC_review_spat}--\eqref{VConstr_review_spat} inherits from its Lagrangian counterpart  \eqref{VC_review}--\eqref{VConstr_review} the same structure, in which the variational constraint \eqref{VConstr_review_spat} follows from the phenomenological constraint  \eqref{PC_review_spat} by formally replacing the time rate of changes by $\delta$-variation, now in the Eulerian setting, see \cite{gay2017lagrangian,gay2017lagrangian2,GBYo2019}. This structure will be used below as a modelling tool for multi-components porous media. 
\end{remark} 
\color{black}

\section{Variational modeling of multicomponent porous media thermodynamics with breaking elastic matrix}
\label{sec:Thermo_comp}

In this Section we derive the general equations of evolution for the thermodynamics of multicomponent porous media with a breaking solid component, by using an extension of the variational formulation recalled above for one-component media. This extension includes the irreversible processes associated to the transitions between the components. We then focus on the particular Lagrangian \eqref{Lagr_def} and discuss the energy and entropy balances. The form of the equations resulting from the variational formulation guides us in the search of possible phenomenological relations for the thermodynamic fluxes compatible with the second law.

We take into account of the irreversible processes of heat exchange between components, the process of transitions between the two elastic components, as well as the friction forces and friction stresses between the components. We assume that each of the transitions between the components generates or absorbs heat. For example, break-up of strands will release heat. For simplicity, we do not include the heat conduction within each component but it can be easily included as done for a one component continua in Section \ref{sec:review_heat_conduction}.

\subsection{Variational formulation and general equations} 

We denote by $ \boldsymbol{f} _k$ and $ \boldsymbol{\sigma} _k$, $k=s,b,f$ the friction forces and stresses acting on the medium $k$, and by $J_{k\ell}$, $k\neq\ell$, $k,\ell \in \{f,s,b\}$ the fluxes associated to the heat exchange between the media $k$ an $\ell$. We assume $J _{k\ell} = J_{\ell k} $ for $k\neq \ell$, see \cite{GBYo2019,gay2022variational}. The conversion rates $q_{k\ell}$, with $q_{k\ell}=-q_{\ell k}$, have  been introduced in \S\ref{conservation_mass}.

Following the general theory \cite{gay2017variational}, we introduce an additional set of matter transport variables $w_k$, which are the analogues to the thermal displacements $ \gamma _k$ introduced earlier for heat transport. This allows to use the relation $F_i \dot x^i \leadsto F_i \delta x^i$ when passing from the phenomenological constraint to the variational constraint, see Remarks \ref{rmk_structure} and \ref{rmk_structure_Eulerian}.

The variational formulation is found by deriving the analogue of  \eqref{VC_review_spat}--\eqref{VConstr_review_spat}. For shortness, we denote by bar a collection of similar quantity relating to all the states present in the media, for example, 
$\bar \bu=(\bu_s,\bu_b,\bu_f)$. We shall use the notation $\bar \phi=(\phi_s, \phi_b, \phi_f)$ for shortness even though $\phi_f=1-\phi_s-\phi_b$. We adapt the notation \eqref{D_t_notation} to  the case of several media as
\begin{equation}\label{D_t_notation_fs}
\begin{aligned} 
&D_t ^kf= \partial _t f + \boldsymbol{u}_k \cdot \nabla f & \qquad  & D_ \delta  ^kf= \delta  f + \boldsymbol{\eta }_k \cdot \nabla f \\
&\bar D_t^k f = \partial _t f + \operatorname{div}(f \boldsymbol{u}_k)  & \qquad  & \bar D_ \delta ^kf = \delta  f + \operatorname{div}(f \boldsymbol{\eta }_k)\,,
\end{aligned}
\end{equation}
where $k=s,b,f$. We arrive at the variational formulation:
\begin{equation}\label{VP_Porousmedia}
\begin{aligned}
&\delta  \int_0^T \Big[\ell( \bar \bu, \bar \rho, \bar s, b ,\bar \phi ) + \sum_k\int_ \mathcal{B} \rho_k  D_t^k  w _k \, {\rm d} ^3  \boldsymbol{x} + \sum_k\int_{ \mathcal{B} } (s _k - \sigma _k  ) D_t ^k  \gamma _k \, {\rm d} ^3  \boldsymbol{x}\Big] {\rm d}t=0\, ,
\end{aligned}
\end{equation} 
subject to the \textit{phenomenological constraints} 
\begin{equation} 
\frac{ \delta \ell }{ \delta  s _k }\bar D_ t ^k  \sigma  _k = \underbrace{\boldsymbol{f} _k \cdot \boldsymbol{u}  _k}_{\stackrel[\text{forces}]{\text{friction}}{}} - \underbrace{\boldsymbol{\sigma} _k : \nabla \boldsymbol{u}  _k}_{\stackrel[\text{stresses}]{\text{friction}}{}}   +\underbrace{\sum_\ell J_{k\ell}D_t^\ell \gamma_\ell}_{\text{heat exchange}}+ \underbrace{\sum_\ell q_{k\ell} D^k_tw_k}_{ \text{matter exchange}}
\label{Phenom_entropy_constr} 
\end{equation} 
and with respect to variations
that satisfy the \textit{variational constraints}
\begin{equation}\label{VarC_Porousmedia}
\begin{aligned} 
& \frac{ \delta \ell }{ \delta  s _k }\bar D_ \delta ^k  \sigma  _k = \underbrace{\boldsymbol{f} _k \cdot \boldsymbol{\eta}  _k}_{\stackrel[\text{forces}]{\text{friction}}{}} - \underbrace{\boldsymbol{\sigma} _k : \nabla \boldsymbol{\eta}  _k}_{\stackrel[\text{stresses}]{\text{friction}}{}}   + \underbrace{\sum_\ell J_{k\ell}D_\delta^\ell \gamma_\ell }_{\text{heat exchange}}+ \underbrace{\sum_\ell q_{k\ell}D^k_ \delta w_k}_{\text{matter exchange}}
\end{aligned} 
\end{equation} 
with $ \delta \gamma _k $, and $\boldsymbol{\eta} _k $ vanishing at $t=0,T$, $k=f,s,b$ and subject to the Euler-Poincar\'e constraints
\begin{equation}\label{EP_variations_thermo} 
\delta \boldsymbol{u}_k= \partial _t \boldsymbol{\eta}_k + \boldsymbol{u}_k \cdot \nabla \boldsymbol{\eta}_k - \boldsymbol{\eta}_k \cdot \nabla \boldsymbol{u}_k,\;\;k=f,s,b, \quad \delta b= - \pounds _ { \boldsymbol{\eta} _s}b\,.
\end{equation}

The forces $\boldsymbol{f} _k$ and stresses $\boldsymbol{\sigma} _k$ are coming from friction and have to be postulated phenomenologically, like $q_{k\ell}$ and $J_{k\ell}$. As we shall see below, the variational principle allows to guide the search for exact forms for these expressions.

The variational formulation \eqref{VP_Porousmedia}--\eqref{VarC_Porousmedia}  yields the general system of equations
\begin{equation}\label{general_thermo} 
\left\{ 
\begin{array}{l}
\displaystyle\vspace{0.2cm} \partial_t \frac{ \delta \ell}{\delta  \boldsymbol{u} _f} + \pounds_{ \boldsymbol{u}_f}  \frac{ \delta \ell}{\delta  \boldsymbol{u} _f}= \rho_f \nabla \frac{\delta \ell }{\delta  \rho _f} + s_f \nabla \frac{\delta  \ell }{\delta  s_f} + \operatorname{div}\boldsymbol{\sigma}_f+ \boldsymbol{f}_f\\
\displaystyle\vspace{0.2cm} \partial_t \frac{ \delta \ell}{\delta  \boldsymbol{u} _b} + \pounds_{ \boldsymbol{u}_b}  \frac{ \delta \ell}{\delta  \boldsymbol{u} _b}= \rho_b \nabla \frac{\delta \ell }{\delta  \rho _b} + s_b \nabla \frac{\delta  \ell }{\delta  s_b} + \operatorname{div}\boldsymbol{\sigma}_b+ \boldsymbol{f}_b\\
\displaystyle\vspace{0.2cm}
\partial_t\frac{ \delta \ell}{\delta  \boldsymbol{u} _s} + \pounds_{ \boldsymbol{u}_s}  \frac{ \delta \ell}{\delta  \boldsymbol{u} _s}= \rho_s \nabla \frac{\delta  \ell }{\delta  \rho _s} + s_s \nabla \frac{\delta  \ell }{\delta  s_s} - \frac{\delta \ell }{\delta b}:\nabla b+  \operatorname{div} \Big( \boldsymbol{\sigma}_s-2 \frac{\delta \ell }{\delta b}\cdot b \Big) + \boldsymbol{f} _s\\
\displaystyle\vspace{0.2cm}\displaystyle\partial_t \rho_k+ \operatorname{div}(\rho_k \boldsymbol{u}_k)=\sum_{\ell} q_{k\ell}\,, \;  k=f,s,b,
\qquad  \partial_t b+ \pounds_{\boldsymbol{u}_s}b=0\\
\vspace{0.2cm}\displaystyle  
\frac{ \delta \ell}{\delta  s _k}  \bar D_t^k s_k = \boldsymbol{f} _k\cdot\boldsymbol{u}_k
 - \boldsymbol{\sigma}_k:\nabla\boldsymbol{u}_k   - \sum_\ell q_{k\ell} \frac{\delta \ell}{\delta \rho  _k} +J_{k\ell}\Big( \frac{ \delta \ell}{\delta  s _k} - \frac{ \delta \ell}{\delta  s _\ell}\Big)\,,\;  k=f,s,b,
\\
\displaystyle\frac{ \delta \ell}{\delta  \phi_k}=0, \;\;  k=s,b,
\end{array}
\right.
\end{equation} 
together with the conditions
\begin{equation} 
\bar D_t^k s _k=\bar D_t^k  \sigma _k +\sum_\ell J_{k\ell} , \quad  D _t ^k\gamma_k= - \frac{\delta \ell }{\delta   s _k }, \quad D _t ^kw_k= - \frac{\delta \ell }{\delta   \rho   _k },
\label{extra_cond} 
\end{equation} 
$k=f,s,b$, which have allowed to eliminate $ \sigma _k $, $ \gamma _k $, and $w_k$ in the final equations, in a similar way to \eqref{two_extra_conditions_fluid}.

In the fluid momentum equations above, we recall that $ \pounds _ {\bu_k} \frac{\delta \ell}{\delta \bu_k}$ denotes the Lie derivative of the fluid momentum density $\frac{\delta \ell}{\delta \bu_k}$ in the direction $\bu_k$, see \S\ref{subsec_Eulerian_descr}. Writing the temperature for each component as 
\begin{equation} 
T_k= - \frac{\delta \ell}{\delta s_k}\, , 
\label{T_k_def}
\end{equation} 
$k=f,s,b$, the equation for the total entropy is deduced as
\begin{equation}\label{total_entropy} 
\begin{aligned}
\sum_k  \bar D_t^k s_k&= - \sum_k \boldsymbol{f}_ k \cdot \frac{ \boldsymbol{u}_k }{T_k}  + \sum_k \boldsymbol{\sigma} _k: \frac{ \nabla \boldsymbol{u}_k }{T_k}+ \sum_{k,\ell} q_{k\ell}\frac{1}{T_k} \frac{\delta \ell}{\delta \rho  _k} + \sum_{k,\ell} \frac{J_{k\ell}}{T_k}(T_k-T_\ell).
\end{aligned}
\end{equation}
Since the system is adiabatically closed, the second law of thermodynamic states that $\sum_k\bar D_t s_k \geq 0$. Based on that requirement, we shall enforce the second law of thermodynamics by choosing appropriate expressions for the forces $\boldsymbol{f}_k$, stresses $\boldsymbol{\sigma}_k$, transition rates $q_{k\ell}$ and fluxes $J_{kl}$ in \eqref{total_entropy}.

\subsection{Porous media thermodynamics with breaking strands}

\subsubsection{Equations of motion}

The Lagrangian for porous media is given by the expression \eqref{Lagr_def}, which we recall here as
\begin{equation}\label{Lagr_def_recall}
\begin{aligned} 
&\ell (\bu_f,\bu_s, \bu_b,   \rho_f,  \rho_s,  \rho  _b, s_f, s_s,  s_b, b, \phi_s, \phi _b) \\
&= \int_{\mathcal{B}}
\Big[\sum_{k=f,b,s}\frac{1}{2}  \rho_k|\bu_k|^2 - \sum_{k=f,b}\rho_k e_k\Big(\bar \rho_k,\frac{s_k}{\rho  _k}\Big)- \rho_s e_s\Big(\bar \rho_s,\frac{s_s}{\rho  _s}, b\Big)\Big] \, {\rm d}^3 \boldsymbol{x}\,,
\end{aligned}
\end{equation}
By using this Lagrangian, the momentum equations, i.e. the first three equations in \eqref{general_thermo}, can be written as 
\begin{equation} 
\label{momentum_full_k} 
\partial _t (\rho  _k  \boldsymbol{u} _k) + \operatorname{div} ( \rho  _k \boldsymbol{u} _k \otimes  \boldsymbol{u} _k)= - \phi _k \nabla p_k + \delta _{ks}\operatorname{div} \boldsymbol{\sigma} _{\rm el}  + \operatorname{div} \boldsymbol{\sigma} _k + \boldsymbol{f}_k\,,
\end{equation} 
$k=f,s,b$, where the pressure $p_k$ and the elastic stress $ \boldsymbol{\sigma} _{\rm el}$ are given by
\[
p_k = \bar \rho  _k ^2 \frac{\partial e_k}{\partial\bar \rho  _k} \quad\text{and}\quad \boldsymbol{\sigma} _{\rm el}= 2 \rho  _s \frac{\partial e_s}{\partial b} \cdot b\,. 
\]
In \eqref{momentum_full_k} the elastic stress appears only for $k=s$ as indicated by the Kronecker symbol $ \delta _{ks}$.

The last equations in \eqref{general_thermo} give the equal pressure conditions
\begin{equation}\label{pressure_condition} 
\frac{\delta \ell}{\delta \phi _s}=0\; \Leftrightarrow \; p_s=p_f \quad\text{and}\quad  \frac{\delta \ell}{\delta \phi _b}=0 \;\Leftrightarrow \; p_b=p_f\,.
\end{equation} 
Hence, writing $p:=p_s=p_b=p_f$ the common pressure, and using the mass balance equations in \eqref{general_thermo}, the momentum equations become 
\begin{equation} 
\label{momentum_reduced_k}
\begin{aligned} 
&\rho_k ( \partial _t \boldsymbol{u} _k + \boldsymbol{u} _k \cdot \nabla \boldsymbol{u} _k)=- \phi _k \nabla p + \delta _{ks} \operatorname{div} \boldsymbol{\sigma} _{\rm el} + \operatorname{div} \boldsymbol{\sigma} _k + \boldsymbol{f}_k- \sum_\ell q_{k\ell} \boldsymbol{u}_k\,,
\end{aligned}
\end{equation}
$k=f,s,b$, where we note the occurrence of the terms $q_{k\ell} \boldsymbol{u}_k$ associated to the transition rates $q_{k\ell}$.   

\color{black} 
Under the usual smoothness assumptions, the formulations \eqref{momentum_full_k} and \eqref{momentum_reduced_k} are completely equivalent, and we shall use the velocity-based formulation \eqref{momentum_reduced_k} in what follows. 
Writing out the equations \eqref{momentum_reduced_k} explicitly for each component,  we get the system
\begin{equation}\label{thermodynamics_lagr_particular}
\left\{
\begin{array}{l}
\vspace{0.2cm}\displaystyle \rho_f (\partial_t \boldsymbol{u}_f+ \boldsymbol{u}_f\cdot \nabla \boldsymbol{u}_f) =  - (1-\phi_s-\phi_b) \nabla p + \operatorname{div}\boldsymbol{\sigma}_f+ \boldsymbol{f}_f 
- \sum_\ell q_{f\ell} \boldsymbol{u}_f\\
\vspace{0.2cm}\displaystyle \rho_b (\partial_t \boldsymbol{u}_b+ \boldsymbol{u}_b\cdot \nabla \boldsymbol{u}_b) = - \phi_b \nabla p + \operatorname{div}\boldsymbol{\sigma}_b+ \boldsymbol{f}_b
- \sum_\ell q_{b\ell} \boldsymbol{u}_b
\\
\vspace{0.2cm}\displaystyle\rho_s( \partial_t \boldsymbol{u}_s+ \boldsymbol{u}_s\cdot \nabla \boldsymbol{u}_s) = - \phi_s \nabla p  +  \operatorname{div} \left( \boldsymbol{\sigma} _{\rm el}  +  \boldsymbol{\sigma}_s \right) +  \boldsymbol{f}_s
- \sum_\ell q_{s\ell} \boldsymbol{u}_s\\
\vspace{0.2cm}\displaystyle\partial_t \rho_k+ \operatorname{div}(\rho_k \boldsymbol{u}_k)=\sum_\ell q_{k\ell}\,, \;\; k=f,s,b,\qquad\partial_t b+ \pounds_{\boldsymbol{u}_s}b=0\\
\vspace{0.2cm}T_k \bar D_t^k s_k = - \boldsymbol{f} _k\cdot\boldsymbol{u}_k
 + \boldsymbol{\sigma}_k:\nabla\boldsymbol{u}_k \\
\vspace{0.2cm}\displaystyle \qquad \qquad + 
 \sum_\ell q_{k\ell} \Big(\frac{1}{2} | \boldsymbol{u}_k| ^2 - g_k\Big) -\sum_\ell J_{k\ell}\left( T_\ell - T_k\right)\,, \;\;  k=f,s,b,\\
\end{array}
\right.
\end{equation}
where we have introduced the notation of $g_k$ for Gibbs' free energy:  
\begin{equation} 
g_k= \pp{}{\rho_k} (\rho_k e_k)=e_k + \bar \rho  _k\frac{\partial e_k}{\partial \bar \rho  _k}  - \eta _k \frac{\partial e_k}{\partial \eta _k}= e_k + \frac{1}{ \bar\rho  _k} p_k  - \eta _kT_k\,.
\label{Gibbs_free_energy_def} 
\end{equation} 

\subsubsection{Energy balances}

It is instructive to write the internal and total energy balance for each component in order to illustrate the impact of each irreversible processes.
Given a specific internal energy function of the form
\[
e_k \left( \frac{ \rho_k }{ \phi _k}, \frac{s_k}{ \rho_k  } ,b \right),
\]
we shall use the general formula
\begin{align*} 
\bar D_t^k ( \rho_k  e_k) &=\left( e_k + \bar \rho  _k\frac{\partial e_k}{\partial\bar \rho  _k} - \eta _k\frac{\partial e_k}{\partial \eta _k}  \right) \bar D^k_t \rho  _k - \bar \rho  ^2 _k \frac{\partial e_k}{\partial\bar \rho  _k}\bar D_t^k \phi _k +\rho  _k \frac{\partial e_k}{\partial b}: D_t^k b + \frac{\partial e_k}{\partial\eta _k}\bar D_t^k s_k\,,
\end{align*}
for the Lagrangian time rate of change of the internal energy density, see \cite{gay2022variational}.

By using the last three equations in \eqref{thermodynamics_lagr_particular} we get the internal energy balance for component $k$ as
\begin{equation}\label{int_energy_k} 
\begin{aligned} 
\bar D_t^k ( \rho  _k e_k) &=- p \bar D^k_t \phi _k + ( \delta _{ks} \boldsymbol{\sigma}_{\rm el} + \boldsymbol{\sigma} _k): \nabla \boldsymbol{u}_k - \boldsymbol{f}_k \cdot \boldsymbol{u}_k \\
& \qquad \qquad \qquad \qquad \qquad + \sum_\ell q_{k\ell} \frac{1}{2} | \boldsymbol{u}_k| ^2 - \sum_\ell J_{k\ell}(T_\ell- T_k)\,,
\end{aligned}
\end{equation} 
where we note that the terms involving the Gibbs free energy have been cancelled out when passing from the entropy to the internal energy balance for component $k$.
The momentum equations in \eqref{thermodynamics_lagr_particular} yields the kinetic energy balance
\begin{equation}\label{E_kin_k} 
\bar D_t^k \Big( \frac{1}{2} \rho  _k| \boldsymbol{u}_k|^2 \Big)  = (- \phi _k \nabla p + \operatorname{div} \boldsymbol{\sigma} _k + \delta _{ks} \operatorname{div} \boldsymbol{\sigma}_{\rm el}+ \boldsymbol{f}_k) \cdot \boldsymbol{u}_k \, ,   
\end{equation} 
so that the total energy balance for component $k$ is
\begin{equation}\label{tot_energy_k} 
\begin{aligned} 
\bar D_t^k \Big(  \frac{1}{2} \rho  _k| \boldsymbol{u}_k| ^2 + \rho  _k e_k \Big) & = \operatorname{div} \big( ( - \phi _k p \delta + \boldsymbol{\sigma} _k + \delta _{ks} \boldsymbol{\sigma} _{\rm el}) \cdot \boldsymbol{u}_k \big) \\
& \qquad \qquad \qquad \qquad - p \partial _t \phi _k - \sum_{\ell} J_{k\ell}( T_\ell- T_k)\,.
\end{aligned} 
\end{equation} 
We note that the contribution of the friction forces $ \boldsymbol{f}_k$ and transition rates $q_{k\ell}$ have been cancelled out, showing that these quantities are only involved in the conversion of kinetic and internal energy within each component. Finally, the total energy balance follows as
\begin{equation}\label{tot_energy} 
\sum_k \bar D_t^k \Big(  \frac{1}{2} \rho  _k| \boldsymbol{u}_k| ^2 + \rho  _k e_k \Big)  = \operatorname{div}\big( \sum_k (- \phi _k p \delta + \boldsymbol{\sigma} _k ) \cdot \boldsymbol{u}_k + \boldsymbol{\sigma} _{\rm el} \cdot \boldsymbol{u}_s  \big)\, ,
\end{equation} 
where we used $\sum_k p\partial _t \phi _k=0$ and $\sum_{k,\ell}J_{k\ell}(T_\ell- T_k) =0$.

Defining the total energy of component $k$ as 
\[
E_k= \int_ \mathcal{B} \Big[\frac{1}{2} \rho  _k| \boldsymbol{u}_k| ^2 + \rho  _k e_k \Big]{\rm d} ^3 \boldsymbol{x}\,,
\]
and using \eqref{tot_energy_k} and \eqref{tot_energy} together with the boundary conditions $ \boldsymbol{u}_k|_{ \partial \mathcal{B}} =0$, yields
\begin{equation}\label{dot_E_k} 
\frac{d}{dt}   E_k= -\int_ \mathcal{B} p \partial _t \phi _k {\rm d} ^3 \boldsymbol{x} +\sum_{\ell}\underbrace{\int_ \mathcal{B} J_{k\ell}( T_k - T_\ell) {\rm d} ^3 \boldsymbol{x}}_{=:P^{\ell \rightarrow k}_{\rm ex}} \, , \;\; \frac{d}{dt} \sum_k E _k =0\,,
\end{equation} 
which allows the identification of the power $P^{\ell \rightarrow k}_{\rm ex}$ transferred from component $\ell$ to component $k$. 
As we see from the expression \eqref{dot_E_k}, the power transfer between different  components involves the heat transfer terms only, which are proportional to the difference in temperature. The friction forces  $ \boldsymbol{f}_k$ do not contribute to this power exchange. Physically, this interesting effect can be explained as follows. When  adjacent microscopic particles from different components $k$ and $\ell$ move with different speeds in the immediate vicinity, they generate friction forces $\boldsymbol{f} _{k \ell}$ that are equal and opposite for both components. 
Each force of friction $\boldsymbol{f}_k= \sum_\ell \boldsymbol{f}_{k \ell} $ is applied only to component $k$, and thus only contributes to the heat in that particular components, see \eqref{int_energy_k}.  It is only after the heat is generated by the friction forces in a given components and absorbed by the current components, that it may be exchanged to other components.

\subsubsection{Thermodynamic considerations}\label{phenomenology}

We shall now proceed with enforcing the second law of thermodynamics by choosing appropriate expressions for the forcing $\boldsymbol{f}_k$, stresses $\boldsymbol{\sigma}_k$, transition rates $q_{k\ell}$ and fluxes $J_{kl}$. Unlike the two-component porous media containing the solid and the fluid, we now have three components: the solid $s$ and the broken $b$ states fully embedded into the fluid $f$.

The friction between different components start acting when there is a discrepancy in velocity in adjacent microscopic particles from different components. We shall only consider two types of friction action. One type of action is due to the forces resulting in the difference in microscopic velocities, denoted as the friction forces $ \boldsymbol{f}$. The other type of friction appears from the difference in velocity gradients, which we call friction stresses $\boldsymbol{\sigma}$.
The friction force $ \boldsymbol{f}_k$ acting on the component $k$ can be written as $\boldsymbol{f}_k= \sum_\ell \boldsymbol{f}  _{k\ell}$
with $\boldsymbol{f}  _{k\ell}= - \boldsymbol{f}_{\ell k}$, similarly with the stresses, i.e., $ \boldsymbol{\sigma}_k= \sum_\ell \boldsymbol{\sigma}  _{k\ell}$ with $ \boldsymbol{\sigma} _{k\ell}=- \boldsymbol{\sigma} _{\ell k}$. Note that viscosity can be included by writing $\boldsymbol{\sigma} _k =  \boldsymbol{\sigma} ^{\rm visc}_ k + \sum_\ell \boldsymbol{\sigma} _{k\ell}$ with $\boldsymbol{\sigma} ^{\rm visc}_ k$ the viscous stress but we shall not consider it. Based on this and using $q_{k\ell}=- q_{\ell k}$ and $J_{k\ell}= J_{\ell k}$, the total entropy equation \eqref{total_entropy} can be written as
\begin{equation}\label{total_entropy_particular} 
\begin{aligned} 
\sum_k  \bar D_t^k s_k &= - \sum_{k<\ell} \boldsymbol{f}_{k\ell} \cdot \Big( \frac{ \boldsymbol{u}_k }{T_k}-\frac{ \boldsymbol{u}_\ell }{T_\ell}\Big)  + \sum_{k<\ell} \boldsymbol{\sigma} _{k\ell}: \Big(\frac{ \nabla \boldsymbol{u}_k }{T_k}- \frac{ \nabla \boldsymbol{u}_\ell }{T_\ell}\Big) \\
& \quad + \sum_{k <\ell} q_{k\ell}\Big( \frac{\frac{1}{2} | \boldsymbol{u}_k| ^2 - g_k}{T_k}-\frac{\frac{1}{2} | \boldsymbol{u}_\ell| ^2 - g_\ell}{T_\ell} \Big)\\
& \quad + \sum_{k<\ell} J_{k\ell}\Big(\frac{1}{T_\ell} - \frac{1}{T_k}\Big)(T_\ell-T_k)\,.
\end{aligned}
\end{equation}

From the first term in \eqref{total_entropy_particular}, it is natural to posit 
\color{black} 
\begin{equation}
\begin{aligned} 
\boldsymbol{f} _{f} & = \mathbb{D}_{fs} 
\left( \frac{\boldsymbol{u} _s}{T_s} -\frac{\boldsymbol{u} _f}{T_f} \right) + 
\mathbb{D}_{fb}
\left( \frac{\boldsymbol{u} _b}{T_b} -\frac{\boldsymbol{u} _f}{T_f} \right) 
\\
\boldsymbol{f} _{b} & = \mathbb{D}_{fb} 
\left( \frac{\boldsymbol{u} _f}{T_f} -\frac{\boldsymbol{u} _b}{T_b} \right) + 
\mathbb{D}_{bs}
\left( \frac{\boldsymbol{u} _s}{T_s} -\frac{\boldsymbol{u} _b}{T_b} \right) 
\\
\boldsymbol{f} _{s} & = \mathbb{D}_{fs} 
\left( \frac{\boldsymbol{u} _f}{T_f} -\frac{\boldsymbol{u} _s}{T_s} \right) + 
\mathbb{D}_{bs}
\left( \frac{\boldsymbol{u} _b}{T_b} -\frac{\boldsymbol{u} _s}{T_s} \right),
\end{aligned}
\label{force_fsb}
\end{equation} 
with $\mathbb{D}_{fs}$, $\mathbb{D}_{fb}$ and $\mathbb{D}_{bs}$ three positive definite, symmetric tensors. Expressions \eqref{force_fsb} are three-component extensions of expressions for the Darcy's law of friction derived in \cite{gay2022variational}. 

From the second term in \eqref{total_entropy_particular}, we posit, similarly to  \cite{gay2022variational},
\begin{equation} 
\label{sigma_expression} 
\begin{aligned} 
\boldsymbol{\sigma}_{f}&=\mathbb{C}_{fs} : \left( \frac{1}{T_f}  \nabla \bu_f-\frac{1}{T_s}  \nabla \bu_s\right) + 
\mathbb{C}_{fb} : \left( \frac{1}{T_f}  \nabla \bu_f-\frac{1}{T_b}  \nabla \bu_b\right) 
\\
\boldsymbol{\sigma}_{b}&=\mathbb{C}_{fb} : \left( \frac{1}{T_b}  \nabla \bu_b-\frac{1}{T_f}  \nabla \bu_f\right) + 
\mathbb{C}_{bs} : \left( \frac{1}{T_b}  \nabla \bu_b-\frac{1}{T_s}  \nabla \bu_s\right) 
\\
\boldsymbol{\sigma}_{s}&=\mathbb{C}_{fs} : \left( \frac{1}{T_s}  \nabla \bu_s-\frac{1}{T_f}  \nabla \bu_f\right) + 
\mathbb{C}_{bs} : \left( \frac{1}{T_s}  \nabla \bu_s-\frac{1}{T_b}  \nabla \bu_b\right) ,
\end{aligned} 
\end{equation} 
where $\mathbb{C}_{fs}$, $\mathbb{C}_{fb}$, and $\mathbb{C}_{bs}$ are positive $(2,2)$ type tensors in the sense that for all matrices $\mathbb{F}$ with components $\mathbb{F} ^i_j$ we have 
\begin{equation} 
\label{pos_def_C}
\mathbb{F}: \mathbb{C}: \mathbb{F} = C^{jp}_{iq} \,\mathbb{F}  ^i_j \, \mathbb{F} ^q_p \geq 0  \quad (\mbox{summation over repeated indices}) \, . 
\end{equation} 
In \cite{gay2022variational} it was shown how the tensors $\mathbb{C}_{k\ell}$ lead to Navier-Stokes like terms in the friction stresses based on the assumption of isotropy and uniformity of tensors in 3D space. 
For each tensor $\mathbb{C}_{k\ell}$ we can use the isotropic+deviatoric+skew-symmetric decomposition to obtain 
\begin{equation} 
\label{C_tensor_example_2} 
\begin{aligned} 
\!\!\!\boldsymbol{\sigma}_{k\ell} &=  \zeta_{k\ell}  \Big( \frac{1}{T_k} \operatorname{div} \mathbf{u} _k\!-\!  \frac{1}{T_\ell} \operatorname{div} \mathbf{u} _\ell \Big) {\rm Id}\! +\! 2\mu_{k\ell}  \Big( \frac{1}{T_k}   \mathbb{F} _k ^{\,\circ}\! -\! \frac{1}{T_\ell} \mathbb{F}  _l ^{\,\circ} \Big)\! + \!\nu_{k\ell} \Big( \frac{1}{T_k}   \mathbb{A}_k \!-
\!\frac{1}{T_\ell}   \mathbb{A} _\ell  \Big) \, , 
\end{aligned} 
\end{equation} 
where $\mathbb{F}_{k}$ and $\mathbb{A}_{k}$ are the symmetric and antisymmetric parts of the velocity gradients,  and $\mathbb{F} _{k}^{\,\circ}$ is the deviatoric (traceless) part of the symmetrized velocity: 
\begin{equation} 
\mathbb{F}_{k} = \frac{1}{2} \big( \nabla \bu_{k} +\nabla  \bu_{k}^\mathsf{T} \big) \, , \quad  \mathbb{A}_{k} = \frac{1}{2} \big( \nabla \bu_{k} -\nabla  \bu_{k}^\mathsf{T} \big) \, , \quad 
\mathbb{F}_{k}^{\,\circ}    = \mathbb{F}_{k}- \frac{1}{3} {\rm Id} \, \operatorname{tr} \mathbb{F}_{k} \, . 
\label{vel_grad} 
\end{equation} 
The conditions on the parameters $\mu , \zeta , \nu$ are
\begin{equation} 
\mu_{k\ell} \geq 0 \, ,   \quad \zeta_{k\ell}  \geq 0 \, , \quad \nu_{k\ell} \geq 0 \, . 
\label{param_cond} 
\end{equation} 
These equations, as demonstrated in \cite{gay2022variational}, when applied to the classical two-components porous media, generalize the Darcy-Brinkman laws of porous media \cite{brinkman1949calculation,brinkman1949permeability,brinkman1952viscosity}, see also \cite{kannan2008flow}.

Note that a more general linear relationship between the gradients $\nabla \mathbf{u}_\ell/T_\ell$ and stresses $ \boldsymbol{\sigma} _\ell$ can also be derived, which is not based on writing $ \boldsymbol{\sigma} _k= \sum_\ell \boldsymbol{\sigma} _{k\ell}$. From \eqref{total_entropy} we can also posit
\begin{equation} 
\begin{bmatrix}
\vspace{0.2cm}\boldsymbol{\sigma} _f\\
\vspace{0.2cm}\boldsymbol{\sigma} _b \\
\boldsymbol{\sigma} _s
\end{bmatrix}
= 
\mathbb{L}
\begin{bmatrix}
\vspace{0.1cm}\frac{\nabla \boldsymbol{u}_f}{T_f} \\
\vspace{0.1cm}\frac{\nabla \boldsymbol{u}_b}{T_b} \\
\frac{\nabla \boldsymbol{u}_s}{T_s} 
\end{bmatrix}, \qquad \mathbb{L}=\begin{bmatrix}
\mathbb{A}_{k\ell}
\end{bmatrix},
\label{lin_relation}
\end{equation} 
where the (3,3) type tensors $\mathbb{A}_{k\ell}$ are such that $ \mathbb{L}$ is positive and such that $ \boldsymbol{\sigma} _f+ \boldsymbol{\sigma} _s+ \boldsymbol{\sigma} _b$ is symmetric, as required by general considerations for the symmetry of stress tensors  \cite{atkin1976continuum}.
We shall not focus on such general expressions in this paper, although they may be relevant for some physical cases. 

\rem{ 
The natural generalization of \eqref{vel_grad} in this case is found as
\begin{equation} 
\label{sigma_fs_general} 
\begin{aligned} 
\boldsymbol{\sigma} _f&= \zeta _{ff} \frac{\operatorname{div} \boldsymbol{u}_f}{T_f}{\rm Id} +   2 \mu _{ff} \frac{\mathbb{F} ^{\, \circ }_f}{T_f}  + \zeta _{fs} \frac{\operatorname{div} \boldsymbol{u}_s}{T_s}{\rm Id} +   2 \mu _{fs} \frac{\mathbb{F} ^{\, \circ }_s}{T_s}  +   \nu \left( \frac{1}{T_f}   \mathbb{A}_f -\frac{1}{T_s}   \mathbb{A} _s  \right)\\
\boldsymbol{\sigma} _s&= \zeta _{sf} \frac{\operatorname{div} \boldsymbol{u}_f}{T_f}{\rm Id} +   2 \mu _{sf} \frac{\mathbb{F} ^{\, \circ }_f}{T_f}  + \zeta _{ss} \frac{\operatorname{div} \boldsymbol{u}_s}{T_s}{\rm Id} +   2 \mu _{ss} \frac{\mathbb{F} ^{\, \circ }_s}{T_s}  -  \nu \left( \frac{1}{T_f}   \mathbb{A}_f -\frac{1}{T_s}   \mathbb{A} _s  \right)\,,
\end{aligned} 
\end{equation} 
with conditions on parameters 
\begin{equation} 
\label{mu_zeta_cond}
\mu _{ff}\geq 0, \quad 4 \mu _{ff} \mu _{ss} \geq ( \mu _{fs}+ \mu _{sf}) ^2 , \quad \zeta  _{ff}\geq 0, \quad 4 \zeta _{ff} \zeta _{ss} \geq ( \zeta _{fs}+ \zeta _{sf}) ^2 , \quad \nu \geq 0\,. 
\end{equation} 
Equations \eqref{sigma_fs_general} with conditions \eqref{mu_zeta_cond} define the most general, thermodynamically consistent choice for the stress tensors. The nature of the parameter $\mu_{kl}$, $\zeta_{kl}$, and $ \nu $ will have to be determined experimentally.
} 
We shall now turn our attention to the terms proportional to $q_{k\ell}$ and $J_{k\ell}$, i.e., the last two terms in \eqref{total_entropy_particular}. Since we assume that no elastic strands are dissolved in the fluid or created from the fluid, we have $q_{fs}=q_{fb}=0$, so that the only nonzero transition rate is $q_{bs}=-q_{sb}$, see \S\ref{conservation_mass}. The second law applied to the corresponding terms thus implies the condition
\begin{equation}\label{positivity_matter_heat} 
q_{sb}\Big(\frac{\frac{1}{2} | \boldsymbol{u}_s| ^2 - g_s }{T_s}-\frac{\frac{1}{2} | \boldsymbol{u}_b| ^2 - g_b }{T_b}\Big) + \sum_{k<\ell}J_{k\ell} \Big( \frac{1}{T_\ell} - \frac{1}{T_k}  \Big)(T_\ell- T_k)\geq 0.
\end{equation}
This condition clearly holds if $J_{k\ell}$ are state functions with
\begin{equation} \label{heat_flux_cond}  
J_{k\ell}\leq 0 \, , 
\end{equation}
and if we posit
\begin{equation}\label{q_cond_single}
q_{sb}=\lambda \Big(\frac{\frac{1}{2} | \boldsymbol{u}_s| ^2 - g_s }{T_s}-\frac{\frac{1}{2} | \boldsymbol{u}_b| ^2 - g_b }{T_b}\Big) \, ,
\end{equation}
for some state function $ \lambda \geq 0$. By state function, we mean a function of the state variables, possibly including velocities. 
Note that \eqref{heat_flux_cond} can be equivalently written as
\begin{equation}\label{heat_flux_cond_2} 
J_{k\ell} \Big( \frac{1}{T_\ell} - \frac{1}{T_k}  \Big)= L_{k\ell} (T_\ell- T_k)
\end{equation} 
for some state function $L_{k\ell} \geq 0$. Conditions \eqref{heat_flux_cond}/\eqref{heat_flux_cond_2} and \eqref{q_cond_single} are discrete analogues to the Fourier and Fick laws. For a general Lagrangian $\ell$, the relation \eqref{q_cond_single} becomes 
\begin{equation}
q_{sb}=\lambda \Big( \frac{1}{T_s} \dede{\ell}{\rho_s} - \frac{1}{T_b} \dede{\ell}{\rho_b} \Big) \,.
\label{general_q_sb} 
\end{equation}

It should be noted that the approach described here derives thermodynamically consistent expressions for the overall transition rates $q_{k\ell}$ between the components $k$ and $\ell$, which is free from any sign constraints. It does not attempt to determine the detailed transition rates $q_{k\ell}^+\geq 0$ and $q_{k\ell}^-\geq 0$ in $q_{k\ell}=q_{k\ell}^+- q_{k\ell}^-$, see \S\ref{conservation_mass}. While such a detailed description is possible in general, we shall consider below the particular case $q_{k\ell}=q_{k\ell}^+-q_{k\ell}^-$ with $q_{k\ell}^-=0$, which leads to a modified version of \eqref{general_q_sb}.

The form of the entropy equation \eqref{total_entropy_particular} also guides the search for cross-phenomena, such as the one occurring between the scalar processes of mass and heat transfer, as well as bulk viscosity, see Appendix \ref{App_A} for details. The resulting formulas share some similarities with the phenomenological equations for the evolution of damage parameters used previously in the literature, see \emph{e.g.} \cite{lyakhovsky2007damage}.


\paragraph{Discussion of the mass transfer rate condition \eqref{q_cond_single}. }Let us consider the condition for the rate of breaking the elastic matrix \eqref{q_cond_single} in more detail using some concrete examples. For the thermodynamic description of the solid, one is usually given the expression of the free energy $ \psi _s( \rho  , T, b)$, rather than the internal energy $e_s( \rho  ,  \eta , b)$, since one usually expresses  the properties  of the solid in terms of the temperature $T$ rather than the specific entropy $ \eta $.

Let us assume a specific free energy of the form
\begin{equation}\label{free_energy_example} 
\psi_s ( \rho  , T, b)= f( \rho  , T) \varepsilon (b)\,, 
\end{equation}
where $ \varepsilon (b)$ is some elasticity energy function and $f( \rho  , T)$ is a coefficient depending on mass density and temperature. Then, the specific entropy is given by 
\[
\eta ( \rho  , T, b)= - \frac{\partial \psi_s }{\partial T}( \rho  , T, b)= - \frac{\partial f}{\partial T} ( \rho  , T) \varepsilon  (b)\,.
\]
This function can be inverted to give the temperature as a function $T=T(\rho  , \eta , b)$. Hence, we get the specific internal energy function for the solid, to be used in the Lagrangian, in the form 
\begin{equation} 
\label{es_expression} 
e _s( \rho  , \eta , b)= \psi _s( \rho  , T, b) + T \eta  \quad \text{with} \quad T=T( \rho  , \eta , b)\,.
\end{equation}
The Gibbs energy, see \eqref{Gibbs_free_energy_def}, is found as 
\begin{equation} 
g_s= \psi _s+ \frac{p}{ \rho  } = f( \rho  , T) \varepsilon (b)+ \frac{p}{ \rho  }   \, , 
\label{Gibbs_free_energy_simple} 
\end{equation} 
with the pressure $p$ common to all components, see \eqref{pressure_condition}. This expression illustrates how the rate of breaking $q_{sb}$ in \eqref{q_cond_single} explicitly involves the elastic energy function $ \varepsilon (b)$.
We recall that when \eqref{Gibbs_free_energy_simple} is used in \eqref{q_cond_single} then the actual densities $ \bar \rho _k $ must be used, rather than the observed densities $ \rho_k$.

The condition for the matrix to start breaking at all is $q_{sb}\leq 0$, \textit{i.e.},
\begin{equation} 
\frac{\frac{1}{2} | \boldsymbol{u}_s| ^2 - g_s }{T_s} \leq \frac{\frac{1}{2} | \boldsymbol{u}_b| ^2 - g_b }{T_b}\,.
\label{positivity_q_sb}
\end{equation}
Expression \eqref{Gibbs_free_energy_simple} illustrates how the Finger deformation tensor $b$ is involved in  condition \eqref{positivity_q_sb} on the left hand side.

Instead of \eqref{free_energy_example}, one can also consider a free energy of the form $ \psi_s ( \rho  , T, b)= \psi _{s0}( \rho  , T)+ \varepsilon (b)$, from which we get $e_s( \rho  , \eta , b)= e_{s0}( \rho  , \eta ) + \varepsilon  (b)$ and $g_s=g_{s0}+ \varepsilon  (b)$ to be used in the condition \eqref{positivity_q_sb}.

If condition \eqref{positivity_q_sb} is violated, no net break-up is possible, but coagulation can still occur. Thus, in the case of absence of coagulation, the breaking rate is given by a modified version of \eqref{general_q_sb}
\begin{equation}\label{q_sb_alt}
q_{sb}=\left\{ 
\begin{array}{ll}
\displaystyle \lambda \Big( \frac{1}{T_s} \dede{\ell}{\rho_s} - \frac{1}{T_b} \dede{\ell}{\rho_b} \Big)\, ,&  
\vspace{0.1cm}\displaystyle \text{if}\;\;
\frac{1}{T_s} \dede{\ell}{\rho_s} \leq \frac{1}{T_b} \dede{\ell}{\rho_b} 
\\ 
\hspace{3mm} 
\displaystyle 0 &  \displaystyle \text{if}\;\;\frac{1}{T_s} \dede{\ell}{\rho_s} \geq \frac{1}{T_b} \dede{\ell}{\rho_b} 
\end{array} 
\right. 
\end{equation}
which also satisfies the second law, see condition \eqref{positivity_matter_heat}. In this case we have $q_{sb}= -q_{sb}^-=- q_{bs}^+$.
The condition \eqref{q_sb_alt} seems to be the only physically relevant case satisfying no healing, \emph{i.e.}, no transfer from broken to solid components. Indeed, if one were to allow even a small transfer of matter from solid to broken phase when the condition \eqref{positivity_q_sb} is violated, it will lead to the violation of the second law of thermodynamics, which is impossible. Thus, the transfer rate must be exactly zero when \eqref{positivity_q_sb} does not hold.

For the case of coagulation, the physics of \eqref{q_cond_single} can be explained as follows. After coagulation, the particles of the broken component fuse together into the solid component and move with the same speed as the solid. The free energy includes the stress tensor $b$ : when the particles fuse, they must immediately start conveying the stress on par with the 'old' solid, as there is no difference between the 'old' and 'new' solid. So the newly formed elastic chains of material particles must immediately deform in such a way that they carry the stress and thus possess the elastic energy. At the same time, we assume that when the particles fuse, they acquire the thermodynamic characteristics of the solid around them, according to the continuum hypothesis. So, it means that there must be enough free energy energy available for the particles to perform that fusion and deformation to immediately fit into the solid matrix.

In the particular applications of breaking under stress we will consider below, we will always be in the domain when \eqref{general_q_sb} is valid and we can thus use the simpler version of breaking rate given by  \eqref{q_cond_single}.

One can also interpret the conditions \eqref{general_q_sb} as generalizing the Karush-Kuhn-Tucker (KKT) conditions for damage mechanics, see \emph{e.g.} \cite{placidi2020variational}. In our case,  condition \eqref{positivity_matter_heat} comes from the second law of thermodynamics guaranteeing the non-decrease in entropy, with $\rho_b$ playing the role of damage parameter. It would be interesting to further explore this connection which we will do in further work. 

\paragraph{Simplification of the equations.} System \eqref{thermodynamics_lagr_particular} simplifies considerably for the case when there is only one transition, namely, the transition from the solid to the broken components discussed above. In that case, since
\begin{equation} 
q_{sb}=-q_{bs}=q\, , \quad 
q_{fs}=q_{sf}=0\, , \quad 
q_{fb}=q_{bf}=0 \, ,
\label{q_constraint_porous} 
\end{equation} 
one gets
\begin{equation}\label{simplified_thermodynamics}
\hspace{-0.1cm} 
\left\{
\begin{array}{l}
\vspace{0.2cm}\displaystyle \rho_f (\partial_t \boldsymbol{u}_f+ \boldsymbol{u}_f\cdot \nabla \boldsymbol{u}_f) =- (1-\phi_s-\phi_b) \nabla p + \operatorname{div}\boldsymbol{\sigma}_f+ \boldsymbol{f}_f \\
\vspace{0.2cm}\displaystyle \rho_b (\partial_t \boldsymbol{u}_b+ \boldsymbol{u}_b\cdot \nabla \boldsymbol{u}_b) = - \phi_b \nabla p + \operatorname{div}\boldsymbol{\sigma}_b+ \boldsymbol{f}_b
-   q \boldsymbol{u}_b
\\
\vspace{0.2cm}\displaystyle\rho_s( \partial_t \boldsymbol{u}_s+ \boldsymbol{u}_s\cdot \nabla \boldsymbol{u}_s) =\displaystyle - \phi_s \nabla p  +  \operatorname{div} \left( \boldsymbol{\sigma} _{\rm el}  +  \boldsymbol{\sigma}_s \right) +  \boldsymbol{f}_s
+ q \boldsymbol{u}_s\\
\vspace{0.2cm}\displaystyle\partial_t \rho_k+ \operatorname{div}(\rho_k \boldsymbol{u}_k)=\pm q,\,  (k=b,s)\, \quad 
\partial_t \rho_f+ \operatorname{div}(\rho_f \boldsymbol{u}_f) =0 
\\
\vspace{0.2cm}\displaystyle
\partial_t b+ \pounds_{\boldsymbol{u}_s}b=0
\\
\displaystyle T_f \bar D_t^f s_f = - \boldsymbol{f} _f\cdot\boldsymbol{u}_f
 + \boldsymbol{\sigma}_f:\nabla\boldsymbol{u}_f -\sum_\ell J_{f\ell}\left( T_\ell - T_f\right)\\
\vspace{0.2cm}T_k \bar D_t^k s_k = - \boldsymbol{f} _k\cdot\boldsymbol{u}_k
 + \boldsymbol{\sigma}_k:\nabla\boldsymbol{u}_k \\
\vspace{0.2cm}\displaystyle \qquad \qquad 
 \pm q \Big(\frac{1}{2}  |\boldsymbol{u}_k| ^2 - g_k\Big) -\sum_l J_{k\ell}\left( T_\ell - T_k\right),\,  (k=b,s)\,,
\end{array}
\right.
\end{equation}
where in each occurrence of $\pm q$, we take the plus sign for $k=s$ and the minus sign for $k=b$.

\begin{remark}[On fluid and solid incompressibility]
{\em When the fluid and/or the solid are incompressible, equations \eqref{thermodynamics_lagr_particular} and \eqref{simplified_thermodynamics} preserve their functional form, although the pressure now acquires a different meaning. In those cases, the equations are derived from the action principle \eqref{VP_Porousmedia} with an added incompressibility condition for the fluid and/or solid, as done in  \cite{FaFGBPu2020,gay2022variational}. The pressure then takes the meaning of the Lagrange multiplier for the incompressibility constraint and not the thermodynamic expression. Additional equations for the equality of pressures follow from the variational principle, and the incompressibility conditions have to be added to close the system \eqref{simplified_thermodynamics}.}
\end{remark}

\rem{
\begin{framed}FGB: for the entropy equation I get
\begin{align*}
T_k \bar D_t^s s_k &= - \boldsymbol{f} _k\cdot\boldsymbol{u}_k
 + \boldsymbol{\sigma}_k:\nabla\boldsymbol{u}_k \textcolor{blue}{+ \sum_\ell q_{k\ell} \left(\frac{1}{2} | \boldsymbol{u}_k| ^2 - g_k\right)} -\sum_\ell J_{k\ell}\left( T_\ell - T_k\right)
\end{align*} 

For the momentum equation, since $\bar D_t^k \rho  _k = \sum_\ell q_{k\ell}$ we get an extra term:
\[
\rho  _k ( \partial _t \boldsymbol{u} _k + \boldsymbol{u} _k \cdot \nabla \boldsymbol{u} _k)= - \phi _k \nabla p + \operatorname{div} \boldsymbol{\sigma} _k^{\rm fr} + \boldsymbol{f}_k- \sum_\ell q_{k\ell} \boldsymbol{u}_k
\]
This extra term is absorbed when the momentum equation is written in terms of $ \rho  _k \boldsymbol{u}_k$:
\[
\partial _t (\rho  _k  \boldsymbol{u} _k) + \operatorname{div} ( \rho  _k \boldsymbol{u} _k \otimes  \boldsymbol{u} _k)= - \phi _k \nabla p + \operatorname{div} \boldsymbol{\sigma} _k^{\rm fr} + \boldsymbol{f}_k.
\]
\textcolor{blue}{VP: Corrected above}
\end{framed} 
\color{black} 
} 

\section{Simplified variational models, breaking and finite time  singularity}\label{sec:1D_motion}

In order to connect our theory to experiments, and provide detailed analytical understanding of the behavior of material under stress, let us consider a certain type of deformations of material, caused by stress in one dimension. Physically, these deformations will connect our theory to several experiments where a uniform block of material is subjected to a deformation in one dimension, normally along the longest axis of the block.  

Mathematically, we will able to obtain analytical solutions and illustrate the occurrence of finite time singularities by making several simplifying assumptions on our model. We will assume that the dynamics can be described by only one velocity and one temperature, and that the kinetic energy can be neglected in the Lagrangian, so that all the quantities evolve slowly and at each time the equation is in a quasi-static equilibrium. We will also explain in details the passing from entropies for each phase to just one effective entropy variable in the variational process for a type of internal energy expression. Finally, we focus on the one-dimensional stretching of a uniform media and exactly solve the equations, whose variational nature and thermodynamic consistentcy is further explored in Sec.~\ref{sec:1D_rip}.

\subsection{Single velocity description}

Let us now consider a simplified model of motion where there is only one characteristic velocity entering the Lagrangian: $\boldsymbol{u}_s=\boldsymbol{u}_b=\boldsymbol{u}_f=\boldsymbol{u}$. Physically, this means that the time scales are large enough for the elastic matrix to drag the fluid with it and for the discrepancies of velocities $\boldsymbol{u}_s$ and $\boldsymbol{u}_f$ to balance out almost instantly. Similarly, when the matrix breaks, the broken parts are slowed down by the fluid to a complete stop on a time scale that is much shorter than the time scale of the relevant dynamics of the elastic media. In that case, the friction forces are acting on very short time scales, but after equilibration, there is no relative motion so the friction part of the bulk force vanishes: $\boldsymbol{f}_k=0$ for $k=s,b,f$. 

Also, during the motion, the heat exchange is so fast that the temperature equilibrates on a much faster time scale than the motion of the elastic media.  Consistently with this, we shall assume that the system can be described by a single entropy. The assumption of the single entropy variable can be justified as follows. If the relative motion of the material is equilibrated quickly and the heat exchange is very fast, $\bu_f=\bu_s = \bu_b=\bu$,  $T_f=T_s=T_b=T$, then the three equations for the rate of change of the entropy  \eqref{simplified_thermodynamics} simplify to
\begin{equation} 
T\bar D_t s_f  =0, \qquad 
T\bar D_t s_s  = q \left( \frac{1}{2} | \boldsymbol{u}| ^2 - g_s \right), \qquad  
T\bar D_t s_b  = -q \left( \frac{1}{2} | \boldsymbol{u}| ^2 - g_b \right)\,.
\end{equation} 
So, defining the total entropy $s:=s_f+s_b+s_s$, we have 
\begin{equation} 
T\bar D_t s= q \left( \left( \frac{1}{2} | \boldsymbol{u}| ^2 - g_s \right) - \left( \frac{1}{2} | \boldsymbol{u}| ^2 - g_b \right) \right) \, . 
\label{Ds} 
\end{equation} 
When $q$ is given by \eqref{q_cond_single}, $T \bar D s \geq 0$ so taking $s=s_b+s_f+s_s$ is thermodynamically consistent. Below, we shall justify the passing from several entropy variables to a single one, for a specific class of Lagrangians.

\paragraph{Variational formulation.} Based on the previous considerations, we shall use the variational formulation for multicomponent fluids with a single velocity and a single entropy but which includes the conversion rates \cite{gay2017variational}. For a given Lagrangian $\ell = \ell( \boldsymbol{u} , \rho_s, \rho_b, s, b)$ and external bulk force $\boldsymbol{f}$ (which may not be arising from friction) it reads as follows: 
\begin{equation}
\label{var_principle_single}
\delta \int_0^T \!\! \left[ \ell( \boldsymbol{u} , \rho  _s, \rho  _b, s, b) +\int_ \mathcal{B}  \left(  \rho  _s D_tw_s+ \rho  _b D_t w_b  \right) {\rm d}^3 \boldsymbol{x}  \right]  {\rm d} t  + \int_0^T \!\!\int _ \mathcal{B}  \boldsymbol{f}  \cdot \boldsymbol{\eta} \,  {\rm d}^3  \boldsymbol{x} {\rm d} t=0 
\end{equation} 
subject to the constraints
\begin{equation}\label{var_principle_single_VC}
\begin{aligned} 
& \frac{\delta \ell}{\delta s} \bar D_t s = q_{sb}D_ t w _s + q_{bs}D_t w _b
\\
& \frac{\delta \ell}{\delta s} \bar D_ \delta  s = q_{sb}D_ \delta w _s + q_{bs}D_ \delta w _b\,,
\end{aligned} 
\end{equation}
as well as to the usual Euler-Poincar\'e variational constraints for $ \delta \boldsymbol{u}$ and $ \delta b$, see \eqref{EP_variations_thermo}. 

The principle \eqref{var_principle_single}--\eqref{var_principle_single_VC} gives, for a general Lagrangian, the single velocity model: 
\begin{equation} 
\label{general_1phase} 
\left\{ \, \,
\begin{aligned} 
& ( \partial _t + \pounds _ {\boldsymbol{u}} ) \frac{\delta \ell}{\delta  \boldsymbol{u} } =  \rho  _s \nabla \frac{\delta \ell}{\delta \rho  _s}+ \rho  _b \nabla \frac{\delta \ell}{\delta \rho  _b}+ s \nabla \frac{\delta \ell}{\delta s} 
- \frac{\delta \ell}{\delta b} :\nabla b - 2\operatorname{div} \left( \frac{\delta \ell}{\delta b} \cdot b \right)  + \boldsymbol{f} 
\\
& \partial _t \rho  _s + \operatorname{div}( \rho  _s \boldsymbol{u} )= q_{sb} =q
\\
& \partial _t \rho  _b + \operatorname{div}( \rho  _b \boldsymbol{u} )= q_{bs}=-q
\\
& \partial _t b + \pounds _{ \boldsymbol{u} }b =0
\\
& \frac{\delta \ell}{\delta s}\bar D_t s = q_{sb} \frac{\delta \ell}{\delta \rho  _b}  + q_{bs} \frac{\delta \ell}{\delta \rho  _s}= q \left( \frac{\delta \ell}{\delta \rho  _b} -\frac{\delta \ell}{\delta \rho  _s}\right) . 
\end{aligned} 
\right. 
\end{equation}

\subsection{Quasi-static approximation of dry material}

To simplify matters further, let us assume that the role of fluid in the dynamics is negligible. As we are interested in the dynamics of break-up, and the entropy expressions such as \eqref{Ds} above involve only the broken and solid phases of the material, we assume that one can drop the fluid terms from the Lagrangian. Physically, one can imagine that as the fluid is removed from the material, the friction between the solid and broken phases generated by the mutual motion will increase dramatically whenever there is a discrepancy in velocities.  As the material 'dries out', the velocities of the solid and broken phases then must coincide to remove the large friction forces from equations. 

For further simplicity of calculation and in order to be able to obtain analytical solutions, we shall assume that the kinetic energy of the material can be neglected in the Lagrangian \eqref{Lagr_def_recall}, and concentrate exclusively on the results of action of forces acting on the material. Neglecting kinetic energy in the Lagrangian states that all the quantities evolve slowly and at each time, the equation is in a quasi-static equilibrium.

\paragraph{Simplified Lagrangians.} The Lagrangian is then just the internal energy and we focus on the expression 
\begin{equation}\label{general_epsilon}
\ell (\boldsymbol{u}, \rho  _s, \rho  _b, s, b)  
= - \int_ \mathcal{B}  \Big[ \rho  _b e_t\Big( \bar \rho_b,  \frac{s_b}{\rho  _b}\Big)   
 + \rho  _s e_t\Big(  \bar \rho_s, \frac{s_s}{\rho  _s}\Big) + \rho  _se_{\rm el}(b)  \Big]\mbox{d}^3 \boldsymbol{x},
\end{equation}
where we choose  
$e_b( \bar \rho  _b, s_b/ \rho  _b)= e_t( \bar \rho  _b, s_b/ \rho  _b)$ and $e_s( \bar \rho  _s, s_s/ \rho  _s,b)= e_t( \bar \rho  _s, s_s/ \rho  _s)+ e_{\rm el}(b)$, for some specific thermal energy function $e_t$.
This expression follows since the solid and broken material are assumed to be the same, and thus have the same dependence  of thermal energy on the specific entropy.
We can assume that the microscopic densities $\bar \rho_b$ and $\bar \rho_s$ do not change much during the breaking process and so the dependence of the internal energy on these variables can be neglected, leading to  
\begin{equation} 
\ell (\boldsymbol{u}, \rho  _s, \rho  _b, s, b) = -  \int _ \mathcal{B} \Big[\rho  _b e_t\Big(  \frac{s_b}{\rho  _b}\Big) + \rho  _s e_t\Big(  \frac{s_s}{\rho  _s}\Big) + \rho  _se_{\rm el}(b) \Big]\mbox{d}^3 \boldsymbol{x} \, . 
\label{lagr_case_simple} 
\end{equation} 
The form of elastic energy  states that it is proportional to the remaining number of intact elastic strands (hence proportionality to $\rho_s$), and the deformation of each strand through the function $e_{\rm el}(b)$. In general, the elastic energy $e_{\rm el}$ can also depend on the effective density of the solid $\bar \rho_s$, leading to the thermodynamic pressure. However, for incompressible materials, such dependence is likely to be neglected. More general cases of internal energies in \eqref{lagr_case_simple} can be considered, but this simple case yields sufficient physical insights without additional mathematical complications. 

It remains to explain how the integrand of the Lagrangian can be expressed in terms of the total entropy density $s$. This step will use the assumption $T_s=T_b=T$ as well as the conservation of the total mass density $ \rho  _s+ \rho  _b$. The temperatures of the solid and broken phases are written as $T_k = e_t'(\eta_k)$ for $k=b,s$, where $\eta_k=s_k/\rho_k$ denotes, as before, the specific entropy of the corresponding phase. From the equality of temperatures $T_s=T_b=T$ we thus get the equality of specific entropies $\eta_s=\eta_b=\eta$.  
If we assume that the initial density of material is constant in space, the total density of the material is conserved in the following sense: 
\begin{equation} 
\rho  _s+ \rho  _b= \varrho _0 J(b) \, , \quad J(b) :=\operatorname{det}( \nabla \boldsymbol{\varphi}  ) ^{-1} \circ \varphi ^{-1}=  (\operatorname{det} b) ^{-1/2}. 
\label{const_density} 
\end{equation}
Hence, the internal energy \eqref{lagr_case_simple} simplifies further to give 
\begin{equation}\label{final_ell}
\ell ( \boldsymbol{u}, \rho  _s, \rho  _b, s, b)  = - \int _ \mathcal{B} \Big[\varrho_0 J(b) e_t( \eta)   + \rho  _se_{\rm el}(b)\Big] \mbox{d} ^3 \boldsymbol{x} \,. 
\end{equation}
Now, we can write the total entropy as $s=s_b+s_s= \eta (\rho_b+\rho_s) = \eta \varrho_0 J(b)$, which allows to express $ \eta $ in terms of the variables $s$ and $b$ in the Lagrangian \eqref{final_ell}.

\paragraph{Equations of motion.} The equations of motion are deduced from the general system \eqref{general_1phase}. Denoting by $\epsilon(\rho_s,b,s)$ the integrand of the Lagrangian \eqref{final_ell}, system \eqref{general_1phase} reduces to \color{black} 
\begin{equation} 
\left\{
\begin{array}{l}
\displaystyle\vspace{0.2cm}0 = - \nabla \left( \rho  _s \frac{\partial \epsilon }{\partial \rho  _s } +  s \frac{\partial \epsilon }{\partial s } -  \epsilon \right) + 2 \operatorname{div}\left( \frac{\partial \epsilon }{\partial b} \cdot b \right) + \boldsymbol{f} \\
\displaystyle\vspace{0.2cm}\partial _t \rho  _s + \operatorname{div}( \rho  _s \boldsymbol{u} )= q, \qquad \partial _t \rho  _b + \operatorname{div}( \rho  _b \boldsymbol{u} )= -q\\
\displaystyle\vspace{0.2cm}\partial _t b + \pounds _{ \boldsymbol{u} }b =0, \qquad  T \bar D_t s = - q    \frac{\partial \epsilon }{\partial \rho  _s}   \, ,
\end{array}
\right. 
\label{simple_lagr_result}
\end{equation} 
with $T=\frac{\partial \epsilon }{\partial s}$ the temperature. From this, and in accordance with the second law, we choose $q = - \frac{ \lambda }{T}\frac{\partial \epsilon }{\partial \rho  _s}$, with $\lambda \geq 0$. In \eqref{simple_lagr_result}, $ \boldsymbol{f} $ is an external force acting on the bulk of the solid material. We shall consider only the boundary forces and thus put $ \boldsymbol{f}=0$ in further considerations. 

It is interesting that for the general type of potential energies depending of specific entropies given by \eqref{lagr_case_simple}, the stress equations given by \eqref{simple_lagr_result} simplify considerably, as we shall see. We use the following lemma.

\begin{lemma}[On derivatives of the Jacobian $J(b)$]
{\rm 
For $J(b)$ defined in \eqref{const_density},
\begin{equation} 
2 \pp{J}{b} \cdot b =  - J(b) {\rm Id}_{3\times3} \, , 
\label{Jacobian_deriv} 
\end{equation} 
where ${\rm Id}_{3\times3}$ is the $3 \times 3$ identity matrix. 
}
\end{lemma}
\begin{proof} We differentiate the Jacobian using the standard rules of the differentiation of determinants: 
\begin{equation} 
J(b)=\frac{1}{\sqrt{\operatorname{det} b}}\, , 
\quad 
\pp{J}{b^{ij}} = -\frac{(\operatorname{adj} b)_{ji}}{2 (\operatorname{det} b)^{3/2} } \, , 
\quad 
\pp{J}{b}=-\frac{\left( b^{-1}\right)^T}{2 \sqrt{\operatorname{det} b} } \, ,
\label{J_deriv_comp_lemma} 
\end{equation} 
since for any invertible matrix $A$, $A^{-1}_{ij} = \operatorname{adj} A_{ij}/\operatorname{det} A$. Multiplying $\pp{J}{b} $ in \eqref{J_deriv_comp_lemma} on the right by $b$ and using the fact that $b$ is a symmetric, the result \eqref{Jacobian_deriv} follows.
\end{proof}

\medskip

Using the result \eqref{Jacobian_deriv} we get the expression $ 2\frac{\partial \epsilon }{\partial b} \cdot b = - ( \rho  _s+ \rho  _b) e_t  \operatorname{Id}_{3 \times 3}+ s e_t'( \eta )  \operatorname{Id}_{3 \times 3}+ 2 \frac{\partial e_{\rm el}}{\partial b} \cdot b$ so the total stress tensor in the first equation \eqref{simple_lagr_result} simplifies considerably: 
\begin{equation} 
\sigma_{\rm tot} = - \left( \rho  _s \frac{\partial \epsilon }{\partial \rho  _s } +  s \frac{\partial \epsilon }{\partial s } -  \epsilon \right) \operatorname{Id}_{3 \times 3} +  2 \frac{\partial \epsilon }{\partial b} \cdot b = 2\rho  _s \frac{\partial e_{\rm el}(b)}{\partial b}\cdot b .
\label{total_stress_no_thermal} 
\end{equation} 
Thus, we arrive at a somewhat surprising physical result that even though  the stress tensor in the dynamics of material does contain terms coming from thermal energy, for a large class of internal energies these terms cancel exactly and only the elastic terms contribute to the actual physically observed stress. The system \eqref{simple_lagr_result} simplifies to
\begin{equation} 
\left\{
\begin{array}{l}
\displaystyle\vspace{0.2cm}0 =   \operatorname{div}\left( \sigma_{\rm el} \right)  \, , \qquad \sigma_{\rm el} = 2 \rho_s \frac{\partial e_{\rm el} }{\partial b} \cdot b \\
\displaystyle\vspace{0.2cm}\partial _t \rho  _s + \operatorname{div}( \rho  _s \boldsymbol{u} )= q, \qquad \partial _t \rho  _b + \operatorname{div}( \rho  _b \boldsymbol{u} )= -q\\
\displaystyle\vspace{0.2cm}\partial _t b + \pounds _{ \boldsymbol{u} }b =0, \qquad   T \bar D_t s = - q\, e_{\rm el}     \,,
\end{array}
\right. 
\label{simple_lagr_result_reduced}
\end{equation}
with the conversion rate $q=- \frac{ \lambda }{T} e_{\rm el}$, $ \lambda \geq 0$.

\begin{remark}{\rm We have regarded the integrand of the Lagrangian in \eqref{final_ell} as a function $ \epsilon ( \rho  _s,b, s)= \varrho _0J(b) e_t( s/ \varrho J(b))+ \rho  _s e_{\rm el}(b)$, where the dependence on $ \rho  _b$ has been entirely expressed in terms of the other variables. However, we can also consider this integrand as a function
$ \epsilon ( \rho  _s, \rho  _b,b, s)= ( \rho  _s+ \rho  _b) e_t( s/ (\rho  _s+ \rho  _b))+ \rho  _s e_{\rm el}(b)$ or $ \epsilon ( \rho  _s,b, s)= \varrho _0J(b) e_t( s/ ( \rho   _b+ \rho  _s))+ \rho  _s e_{\rm el}(b)$, which explicitly includes a dependence on $ \rho  _b$. Applying the equations \eqref{general_1phase} to this alternative dependencies consistently gives the same final equations \eqref{simple_lagr_result_reduced} .
}
\end{remark}

\color{black} 

\subsection{Uniform strain assumption and reduced dynamics}

Let us consider the reduced motion of a material described by \eqref{simple_lagr_result_reduced}, for the case of a uniform block stretched in opposite directions by an external force $F(t)$ applied to the boundaries. We shall assume that all the particles of the broken and solid components are moving along one-dimensional trajectories and neglect the transversal motion, except for the potential energy. In reality, under this deformation, there will also be a transversal component of the velocity, but since we will be neglecting the kinetic energy in the Lagrangian, this velocity component will not play a role in dynamics. The transversal component of the deformation will be incorporated later through an appropriate form of the elastic energy. Since the force applied to both sides of the block $x = \pm L$ is equal and points in the opposite directions,
it is natural to assume that all the deformations are antisymmetric around the origin $x=0$. This assumption naturally leads to the following ansatz:
\begin{equation}
\begin{aligned}
x(t,X) & =  \big( D(t) +1 \big) X \;\; \Longrightarrow \;\; 
u_s(t,x)  =  \frac{\dot D(t)}{D(t)+1} x \\
b (t,x)&=(D+1)^2=b(t) \;\; \Longrightarrow \;\; \sigma_{\rm el}= \sigma_{\rm el}(D(t)) 
\\ 
s (t,x) & = s (t)  \;\; 
\Longrightarrow \;\; 
T (t,x)  = T (t)\phantom{\frac{1}{2} }
\\ 
\rho_s(t,x) & = \rho_s (t) \,.
\end{aligned}
\label{ansatz_linear}
\end{equation}
\rem{ 
\todo{VP: It is a bit confusing because the actual specific energy of the elastic deformation is $e_{\rm el,t}=\varphi(\rho_s)e_{\rm el}(b)$. Don't know if I should redefine the elastic energy or drop $\varphi(\rho_s)$. I like it because it allows us enough generality; but we never actually use it, only use $\varphi(\rho_s)=1$ } 
} 
In the first equation of the above system, we have defined   $x=(D+1)X$ so $D=0$ corresponds to the equilibrium, and thus $D$ is small for small deviations from equilibrium. 

\rem{ 
\todo{\textcolor{magenta}{ FGB: Just to be sure. $X$ is the material point in the fixed reference configuration, say  $X \in  \mathcal{B} =[-1,1]$. Then the strand in space at time $t$ occupies the space $ \mathcal{B} (t)=\left[ \frac{-1}{D(t)+1}, \frac{1}{D(t)+1} \right]$. So positive $D$ means compression, and negative $D$ ($-1<D<0$) means stretching.\\
See my confusion later.}
\\
\textcolor{blue}{I think I know where the confusion comes from. My definition above is awkward. I followed our old paper on tubes with moving water, but it is better to define things in a more straightforward way. I hve refefined the variables as: $x=(D+1)X$, so $\varphi(X,t)=(D(t)+1)x$, and $\varphi^{-1}(X,t) = X/(D+1)$. So I get the reverse result, the strand occupies the domain in space $ \mathcal{B} (t)=\left[  -(D(t)+1),  D(t)+1 \right]$, and $D>0$ is stretching. I have corrected it now. Tensor $b$ is nicer too. Notice that $u$ gets a plus sign which I can fix in the equations below. }
}} 
Upon integration of \eqref{simple_lagr_result_reduced} over the volume of the solid, the terms inside the volume cancel, whereas the terms on the boundary have to match the boundary force $f_b$, applied to the boundary area $A$.  Denoting $F=f_b/A$ for shortness, we have on the boundary 
\begin{equation} 
\sigma_{\rm el} =  \rho_s \pp{e_{\rm el}}{D} (D+1) =F 
\label{applied_force} 
\end{equation}
The stretching force $F$ is prescribed and is specified by the physical apparatus acting on the material.
Using \eqref{ansatz_linear}, system \eqref{simple_lagr_result_reduced} together with \eqref{applied_force} give the coupled algebraic-differential equations 
\begin{equation} 
\left\{ 
\begin{aligned} 
& \dot \rho_s + \frac{\dot D}{D+1} \rho_s = q \ , \qquad q= - \lambda \frac{e_{\rm el}(D)}{T}\, ,
\\
& T \Big( \dot s + \frac{\dot D}{D+1} s \Big) = - q e_{\rm el} (D) \, ,
\\
&  \rho_s \pp{e_{\rm el}}{D} (D+1) =  F \,, \qquad  s = \eta(T) J(D) \varrho _0\, ,
\end{aligned} 
\right. 
\label{rho_T_coupled_modified_new} 
\end{equation} 
where in the last equation we used $s=\eta (\rho_b+\rho_s) = \eta J(D)\varrho _0$ and computed $\eta(T)$ by inverting the relationship $T=e_t'(\eta)$. That \eqref{rho_T_coupled_modified_new} also comes from a variational formulation, in a similar way with \eqref{simple_lagr_result_reduced}, will be shown in \S\ref{sec:1D_rip}.

We will explore below the occurrence of finite-time singularities in system \eqref{rho_T_coupled_modified_new} for two choices of the phenomenological coefficient $ \lambda $, namely $ \lambda = \lambda _0T$ and $ \lambda $ = const. 

\paragraph{Energy conservation.} We notice that the total energy is conserved, as the following calculation shows: 
\begin{equation} 
\begin{aligned} 
\frac{d \epsilon}{d t} &= \frac{d}{dt} \big((\rho_s + \rho_b) e_t(\eta) + \rho_s e_{\rm el}(D) \big) 
\\ 
& 
= - e_t(\eta) ( \rho_s + \rho_b) 
 \frac{\dot D}{D+1}  
 +
  e_{\rm el}(D) \Big( - \frac{\dot D}{D+1}  \rho_s + q \Big)\\
& \qquad + \rho_s \pp{e_{\rm el}}{D} \dot D + (\rho_s+\rho_b) T \dot \eta 
 = - \epsilon \frac{ \dot D}{D+1} 
+F \frac{\dot D}{D+1} \, , 
\end{aligned} 
\label{dedt_calc} 
\end{equation} 
giving the energy balance in the spatial frame
\begin{equation} 
\frac{d \epsilon}{d t} + \frac{\dot D}{D+1} \epsilon =  \frac{\dot D}{D+1} F\,.
\label{cons_energy} 
\end{equation} 
The terms on the left-hand side are the total derivative of the energy, including convection, and the term on the right-hand side is the work 
provided by the external force at the boundary. 
Indeed, that term is simply equal to the product of $F \cdot u$ on the boundary. 
See \S\ref{sec:1D_rip} for additional justification of this form of energy balance.

\paragraph{Functional forms of $e_{\rm el}(D)$ for Ogden materials.} To be concrete, let us consider an Ogden-like material, where the elastic energy $e_{\rm el}(D)$ is expressed in terms of principal stretches. For example, for the $N$-term Ogden material, we have
\begin{equation} 
\label{Ogden_Nterm}
e_{\rm el}(D)= \sum_{k=1}^N \frac{\mu_k}{\alpha_k} \sum_{i=1}^3 \lambda_i^{2 \alpha_k} \, , 
\end{equation} 
see \cite{ogden1972large},  as well as recent papers \cite{mihai2015comparison, mihai2017family, rhodes2022comparing} for application of these models to brain tissue mechanics. 

The principal stretch in 1D is simply $\lambda_1=(D+1)$. For the other two directions for uniform material under stretching we have  $\lambda_2 = \lambda_3$. If we assume, for example, that the material is close to incompressible, and the sides of the block normal to the $y$ and $z$ axes are free, then $\lambda_1 \lambda_2 \lambda_3 =1$ so $\lambda_2=\lambda_3 = \lambda_1^{-1/2}$. 
So, for the one-term Ogden model we get
\begin{equation} 
e_{\rm el}(D)=\frac{\mu}{\alpha} (D+1)^{2 \alpha}+2 \frac{\mu}{\alpha} (D+1)^{-\alpha}\,.
\label{Ogden_1term} 
\end{equation} 
Then, 
\[
e_{\rm el}'(D)=2 \mu  \left[ (D+1)^{2 \alpha}-  (D+1)^{-\alpha} \right] 
\]
and $D=0$ is an equilibrium for arbitrary $\alpha>0$. The internal elastic energy $e_{\rm el}(D)$ has a unique minimum, diverges as $D \rightarrow -1$ and also diverges as $D^p$ as $D \rightarrow + \infty$.

\paragraph{Finite-time singularities and breaking: a simplified case.} 
The equation \eqref{rho_T_coupled_modified_new} form a closed differential-algebraic system that can be solved numerically for any specified thermal energy $e_t(\eta)$.  We are going to be interested in the cases of a catastrophic failure of the bar when $\rho_s \rightarrow 0$. We can take $\lambda = \frac{1}{2} \lambda_0 (T_s + T_b) \simeq \lambda_0 T$
in  \eqref{general_q_sb} which is thermodynamically consistent. Then, $q=- \lambda _0 e_{\rm el}(D)$ and the thermodynamic part of equations \eqref{rho_T_coupled_modified_new} decouples leading to: 
\begin{equation} 
\left\{ 
\begin{aligned} 
& \frac{d}{dt} \left[ \rho_s (D+1)   \right]  =  - \lambda_0 e_{\rm el}(D) (D+1)
\\
&  \rho_s \pp{e_{\rm el}}{D} (D+1) =  F \,.
\end{aligned} 
\right. 
\label{rho_T_coupled_modified2} 
\end{equation} 
This system can be reduced to a single equation for $D(t)$ which is solvable in quadratures: 
\begin{equation} 
\dot D = \frac{\lambda_0}{F}\frac{e_{\rm el}(D) (D+1) \left[ e_{\rm el}'(D)\right]^2 }{e_{\rm el}''(D)} \,.
\label{D_eq_single} 
\end{equation}

For Ogden material described by \eqref{Ogden_1term}, $e_{\rm el}(D)>0$ and $e_{\rm el}''(D)>0$ for $\alpha>1$, so for all extension forces $F>0$, the right-hand side is non-negative. Moreover, $e_{\rm el}'(D)=0$ only for $D=0$, so $D=0$ is the only critical point of the equation \eqref{D_eq_single}. Thus, for all initial conditions $D>0$, all trajectories diverge to $+ \infty$. Moreover, the divergence happens in finite time. Indeed, since 
$e_{\rm el}(D) \sim D^{2 \alpha}$ as $D \rightarrow \infty$, the right-hand side of \eqref{D_eq_single} diverges as $D ^{4 \alpha+1}$, and the singularity of solutions \eqref{D_eq_single} $F>0$ takes the form: 
\begin{equation} 
D \sim \Big( \frac{\lambda_0}{F} (t_*-t)\Big)^{\frac{1}{4 \alpha}} \, . 
\label{singularity_simple} 
\end{equation} 
From the force balance equation, the amount of solid material collapses to zero in finite time according to the square root law: 
\begin{equation} 
\rho_s = F \big[e_{\rm el}'(D) (D+1)\big]^{-1} \sim D^{2 \alpha} \sim \left( \frac{\lambda_0}{F} \right)^{\frac{1}{4 \alpha}}  (t_*-t)^{1/2} \, . 
\label{square_root_density} 
\end{equation} 

\paragraph{Analytical solution, integrals of motion and singularities in the full equation \eqref{rho_T_coupled_modified_new}.} Interestingly, the full equation \eqref{rho_T_coupled_modified_new} can also be solved in quadratures for $\lambda=$ const.  

Suppose we specify $e_t(\eta)$, with $T=e_t'(\eta)$. Inverting this relationship, we have $\eta=\eta(T)$, which we leave in its general form for now. The total entropy $s$ is then derived  as $s = (\rho_s+\rho_b)\eta(T)  = \varrho_0 J(D)  \eta(T)$. Since $J(D)=1/(D+1)$, the equation \eqref{rho_T_coupled_modified_new}  now gives 
\begin{equation} 
\left\{ 
\begin{aligned} 
&    \frac{d}{dt} \left(  \frac{F}{e_{\rm el}'(D)} \right)  = - \lambda \frac{ e_{\rm el}(D)}{ T} (D+1) 
\\
&   \varrho_0 T \eta'(T) \frac{d T}{dt}  =  \lambda \frac{ e_{\rm el}(D)^2}{ T} (D+1) 
\end{aligned} 
\right. 
\label{tot_eqs_particular}
\end{equation}
and $\rho_s(D) $ given by the force balance, \emph{i.e.}, the third equation of \eqref{rho_T_coupled_modified_new}. Combining the equations above, we notice that we can separate variables and connect $D$ and $T$ through 
\begin{equation} 
F\frac{e_{\rm el}''(D) e_{\rm el}(D)}{\left[ e_{\rm el}'(D) \right]^2 } \frac{d D}{d t} = \varrho_0 \eta'(T) T \frac{d T}{d t}\, , 
\label{variable_separate_eqs} 
\end{equation} 
leading to a first integral of the type 
\[
F \Phi(D) = \varrho_0  \Psi(T) + C\, , 
\]
\[
\Phi(D)  := \int \frac{e_{\rm el}''(D) e_{\rm el}(D)}{\left[ e_{\rm el}'(D) \right]^2 } \mbox{d} D \, , \quad 
\Psi(T):= \int T \eta'(T) \mbox{d} T \, . 
\]
If we assume $e_t(\eta) = c e^{\kappa \eta}$, $c$ = const, which also happens to be the most commonly used case for the dependence of internal energy on the specific entropy, we obtain $\Psi(T)=T/ \kappa $, where we note that this expression is exact and valid for all $T$, not just large $T$. 
Then, $T = F \kappa   \Phi(D)/ \varrho_0$, where we absorbed the constant due to the initial conditions in the definition of the function $\Phi(D)$. We thus arrive at the single equation for $D$: 
\begin{equation} 
\dot D = \frac{\lambda \varrho_0}{\kappa F^2} \frac{e_s (D) [e_{\rm el}'(D)]^2 (D+1)}{e_{\rm el}''(D) \Phi(D)}.
\label{single_eq_D_reduced} 
\end{equation} 
Again, assuming Ogden materials described by \eqref{Ogden_1term}, we have $e_{\rm el}(D) \sim D^{2 \alpha}$ and $\Phi(D) \rightarrow D$ as $D \rightarrow \infty$. We thus arrive to the following expressions for the singularity: 
\begin{equation} 
D \sim \Big( \frac{\lambda \varrho_0}{\kappa F^2} (t_*-t) \Big)^{\frac{1}{4 \alpha-1}} \, , \quad 
\rho_s \sim \Big( \frac{\lambda \varrho_0}{\kappa F^2} (t_*-t) \Big)^{\frac{2 \alpha}{4 \alpha-1}} 
\, . 
\label{D_rho_asymptotic} 
\end{equation} 
\paragraph{On experimental verification of the breaking laws \eqref{singularity_simple} and \eqref{D_rho_asymptotic}.} Note that that according to \eqref{D_rho_asymptotic}, $\rho_s \rightarrow 0$ slower than the square root law given by \eqref{singularity_simple}, since the power of $t_*-t$ is always larger than $1/2$. In addition, the dependence of the density on the applied force $F$ is different in \eqref{singularity_simple} and \eqref{D_rho_asymptotic}. This presents an opportunity to experimentally verify the predictions provided in this Section. We are not currently aware of such experiments of applying a \emph{constant} force to break a uniform rod of material. However, there are experiments relating to breaking under periodic forcing. We now consider these experiments and connect them with the theory derived here. 

\color{black} 


\paragraph{Long-term damage due to repeated load.} 
Let us now investigate another limit, namely, a cyclic load experiment that is common in the literature for the experimental analysis of material wear and damage. In particular, let us assume that the applied force $F(t)$ is a periodic function, for example, $F(t)=F_0 \cos \omega t$. Let us also assume that over the  period of the load $\Delta t= 2 \pi/\omega$, the temperature is able to equilibrate, so $T\simeq T_*$. Multiplying the first equation of \eqref{rho_T_coupled_modified_new} by $(D+1)$ and denoting $k \Delta t:=t_k$, the first equation of \eqref{rho_T_coupled_modified_new} gives, at the $k$-th period 
\begin{equation} 
\rho_s(t_{k+1}) (1+D(t_{k+1}))   - \rho_s(t_k)(1+D(t_{k})) 
= \lambda _0
\int_{t_k}^{t_{k+1}} e_{\rm el}(D (t) ) (1+D(t)) \mbox{d} t \,,
\label{Var_periodic_rho}
\end{equation}
where we considered the case $ \lambda = \lambda _0T$ treated earlier.

Because of the quasi-static nature of these equations, the force balance, \emph{i.e.}, the third equation of  \eqref{rho_T_coupled_modified_new}, holds for all times, so the equation \eqref{Var_periodic_rho} above reduced to 
\begin{equation} 
\frac{F}{e_{\rm el}'(D(t_{k+1}))} - \frac{F}{e_{\rm el}'(D(t_{k}))} =- \lambda _0 \int_{t_k}^{t_{k+1}} e_{\rm el}(D (t) ) (1+D(t))  \mbox{d} t \,.
\label{Var_periodic_D} 
\end{equation} 
In experiments, each cycle of forcing leads to small changes in $\Delta D_k = D(t_{k+1})-D(t_k)$. If we divide both sides of \eqref{Var_periodic_D} by $\Delta t$, and use linearization of the left-hand side for small $\Delta D_k$, we arrive to 
\begin{equation} 
\frac{e_s''(D)}{\left[ e_s'(D)\right]^2 } \frac{\Delta D}{\Delta t} \simeq  \frac{\lambda_0}{F}  e_s(D) (1+D)\,, 
\label{Var_periodic_D_eq} 
\end{equation} 
where the variables are now understood in the sense of having the mean value of the true variables over the forcing period. 
The equation \eqref{Var_periodic_D} coincides exactly with the equation \eqref{D_eq_single},
and thus we obtain the power-law singularity for $D(t)$ according to \eqref{singularity_simple}. The density of the solid phase $\rho_s$, according to \eqref{square_root_density}, tends to zero in finite time as a square root $\rho_s \simeq (t_*-t)^{1/2}$, or, in terms of the number of load cycles, $\rho_s \simeq (N_*-N)^{1/2}$.

It would be interesting to compare these predictions of our theory with experimental results. The experiments on acrylic based bone cement  \cite{murphy2000magnitude} and vaccum-mixed bone cement \cite{jeffers2005damage} for biomedical applications utilize the Weibull-style distribution of probability of failure with the number of cycles $P_f \sim e^{-(N_f/b)^a} $, where $N_f$ is the number of fatigue cycles to failure, and $k,m$ are some positive constants fitted to experimental data \cite{weibull1951statistical}. It is possible that our work can contribute to a better understanding of the fatigue behavior because of the nature of the analytical prediction \eqref{singularity_simple}. In accordance with \cite{murphy2000magnitude,jeffers2005damage}, damage accumulation $\rho_b=\rho_0-\rho_s$, with $\rho_s$ described by \eqref{singularity_simple} is a  non-linear function of the number of cycles $N=t_N/\Delta t$, increasing monotonically and exhibiting a singularity at the breakdown point $N=N_*$.

\section{ Finite dimensional reduced variational models}\label{sec:1D_rip}

In this Section we show that one can derive the equations  \eqref{rho_T_coupled_modified_new} directly from the `discrete  variational' approach to the problem, by approximating the quantities characterizing the materials by their typical values. We believe it is valuable to present this derivation in some details, as we have to be careful when describing the quantities in the material or in the spatial representation, even though we are discussing the system in the discrete approximation. This distinction reflects in the form of the time rate of change of scalar and densities, as well as in the form of the energy and entropy balances in both representations, which is fundamental for thermodyamic considerations.
For simplicity, we consider the single velocity approach of Section \ref{sec:1D_motion}, but the variational derivation that we present can be extended to the complete multi velocity model described in Section \ref{sec:Thermo_comp}.

Recall that we focus on one-dimensional linear deformations of the form $ \varphi (t,X)= (D(t)+1)X$. We assume as before that the mass and entropy densities in the material frame do not depend on $X$, \textit{i.e.}, $ \varrho _k(t,X)= \varrho _k(t)$ and $ S(t,X)=S(t)$. Hence, to a standard Lagrangian $L( \varphi , \dot  \varphi , \varrho _k, S)$ of the continuum model we can associate a finite dimensional Lagrangian $L(D, \dot  D, \varrho _k, S)$. We derive below the equations of motion by using finite dimensional versions of the variational formulation of thermodynamics in the material and Eulerian frame, and show how it directly yields to the model in \eqref{rho_T_coupled_modified_new} for an appropriate choice of $L$.

\paragraph{Discrete material representation.} The finite dimensional variational formulation of a thermodynamic media gives \cite{gay2017lagrangian}: 
\begin{equation}\label{material_VP} 
\delta \int_0^T \left[ L(D, \dot  D, \varrho _k, S)+ \varrho _k \dot  W^k \right] {\rm d} t + \int_0^T F^*  \delta D \,{\rm d} t=0
\end{equation} 
subject to the phenomenological and variational constraints
\begin{equation}\label{material_VP_const} 
\begin{aligned} 
& \frac{\partial L}{\partial S} \dot  S= \sum_{k,\ell} Q_{k\ell}\dot  W^k 
\\ 
 & \frac{\partial L}{\partial S} \delta   S= \sum_{k,\ell} Q_{k\ell} \delta   W^k\,,
\end{aligned}
\end{equation}
where $Q_{k\ell}$ are the conversion rates and $F^*$ and external force. It is the finite dimensional analogue to the variational formulation used earlier in the continuum setting. 
This variational principle gives the equations of motion 
\begin{equation} 
\label{eqs_spatial_discrete}
\left\{
\begin{array}{l}
\displaystyle\vspace{0.2cm}\frac{d}{dt} \frac{\partial L}{\partial \dot  D} - \frac{\partial L}{\partial D} = F^*\\ 
\displaystyle\dot  \varrho _k= \sum_\ell Q_{k\ell}\\ 
\displaystyle\frac{\partial L}{\partial S} \dot  S= -\sum_{k,\ell} Q_{k\ell}\frac{\partial L}{\partial \varrho _k} = \sum_{k<\ell} Q_{k\ell} \left( \frac{\partial L}{\partial \varrho  _\ell} - \frac{\partial L}{\partial \varrho _k} \right)\,, 
\end{array} \right. 
\end{equation} 
together with the condition $ -\frac{\partial L}{\partial \varrho _k} = \dot  W^k$.
The balance of total energy, defined for a general Lagangian by $E_{\rm tot}= \frac{\partial L}{\partial \dot  D} \dot  D- L$, follows as 
\begin{equation}\label{energy_balance_1D} 
\frac{d}{dt} E_{\rm tot}= F ^* \dot  D\,.
\end{equation}
For standard continuum theories one gets finite dimensional Lagrangians of the form
\begin{equation}\label{Lagrangian_1D} 
L(D, \dot  D, \varrho _k, S)=   \frac{1}{2}  M \varrho\dot  D ^2  - E(D, \varrho _k, S)\,,
\end{equation}
for $M>0$ and $ \varrho = \sum_k \varrho _k$, while for the quasi-static approximation considered earlier, one chooses $L(D, \dot  D, \varrho _k, S)= - E(D, \varrho _k, S)$.

\paragraph{Discrete spatial representation.} To pass from the material to the spatial representation, we multiply  the densities by $1/(D+1)$, whereas the scalars are transformed to the spatial frame without change: 
\begin{equation} 
\label{spatial_description} 
\begin{aligned} 
& s= \frac{S}{D+1}\,, \qquad \rho  _k= \frac{\varrho_k}{D+1}\,, \qquad q_{k\ell}= \frac{Q_{k\ell}}{D+1}\,, 
\\
&  F^*= F(D+1)\,, \quad 
w ^k = W ^k\,.
\end{aligned} 
\end{equation} 
These relations follow from the relations $s(t,x)= S(t, \varphi ^{-1} (t,x)) J \varphi ^{-1} (t,x)$ for densities and the relations $w^k(t,x)=W^k(t, \varphi ^{-1} (t,x))$ for scalar functions, in the case of linear deformations $ \varphi (t,X)= (D(t)+1)X$ and when $S(t,X)=S(t)$ and $W^k(t,X)=W^k(t)$.

Correspondingly, we have following relations between the time rate of change in the material and spatial frame for densities
\[
\frac{d}{dt} \varrho _k= (D+1)\bar D_t \rho  _k
\, , \quad 
\frac{d}{dt} S= (D+1)\bar D_t s 
\]
and for scalars
\[
\frac{d}{dt} W^k= \frac{d}{dt} w^k=:D_t w^k\, ,
\]
where
\[
D_t= \partial _t \quad\text{and}\quad \bar D_t = \partial _t + A_s \quad\text{with}\quad  A_s =  \frac{\dot  D}{D+1}
\]
are the finite dimensional version of the Lagrangian time derivative of a scalar function and a density, see \eqref{D_t_notation}. 

By changing variables in the variational principle \eqref{material_VP} we get
\begin{equation} 
\label{var_discrete_spatial} 
\delta \int_0^T \! (D+1) \left[\ell(D, \dot  D, \rho  _k, s)+ \rho  _k D_t  w^k \right] {\rm d} t + \int _0^T\!\! (D+1) F \delta D \, {\rm d} t=0
\end{equation} 
with phenomenological and variational constraints
\begin{equation} 
\label{var_discrete_spatial_constr} 
\begin{aligned} 
& \frac{\partial \ell}{\partial s} \bar D_t  s= \sum_{k,\ell} q_{k\ell} D_t  w^k \\ 
& \frac{\partial \ell}{\partial s} \bar  D_ \delta s= \sum_{k,\ell} q_{k\ell} D_ \delta w ^k,
\end{aligned} 
\end{equation}
where we defined the Lagrangian $\ell$ in the spatial frame from the material Lagrangian $L$ by
\begin{equation}\label{L_ell} 
L(D,\dot  D, \varrho _k,S)= (D+1) \ell\big(D, \dot  D,  \varrho _k/(D+1),S/(D+1)\big).
\end{equation} 
Note that the variational formulation \eqref{var_discrete_spatial}--\eqref{var_discrete_spatial_constr} is clearly a finite dimensional version of \eqref{var_principle_single}--\eqref{var_principle_single_VC}.  
It is also important to mention that in the discrete consideration, when all the block of material is considered a single entity, the force $F$ can be understood as either a spatially uniform force applied to the whole material, or the force applied to the boundaries. We remind the reader that in the previous section, the force $F$ was applied to the boundaries of the rod under stretching. 

Application of the principle \eqref{var_discrete_spatial}--\eqref{var_discrete_spatial_constr} gives the equations of motion 
\begin{equation} 
\label{eqs_spatial_discrete_2} 
\left\{
\begin{array}{l}
\displaystyle\vspace{0.2cm} \frac{d}{dt} \frac{\partial \ell}{\partial \dot  D} - \frac{\partial \ell}{\partial D}= \frac{1}{D+1} \Big( - \frac{\partial \ell}{\partial \dot  D} \dot  D -\frac{\partial \ell}{\partial s} s - \sum_k \frac{\partial \ell}{\partial \rho  _k} \rho  _k + \ell \Big) + F\\
\displaystyle\vspace{0.2cm} - \frac{\partial \ell}{\partial s} \bar D_t s = \sum_{k,\ell} q_{k\ell} \frac{\partial \ell}{\partial \rho  _k}\\
\displaystyle \bar D_t \rho  _k = \sum_\ell q_{k\ell} \,,
\end{array}
\right.
\end{equation}
together with the condition $ -\frac{\partial \ell}{\partial \rho _k} = D_t  w^k$. The balance of total energy $ \epsilon _{\rm tot}= \frac{\partial \ell}{\partial \dot  D} \dot  D -\ell$ in the spatial frame for system \eqref{eqs_spatial_discrete_2} reads
\begin{equation} 
\label{energy_cons} 
\bar D_t \epsilon _{\rm tot} = F\dot  D\, \quad \Leftrightarrow \quad 
\frac{d}{dt} \left(  \epsilon _{\rm tot} (D+1)  \right) = (D+1) F\dot  D\,, 
\end{equation} 
\rem{ 
It can be written as
\[
\frac{d}{dt} \left(  \epsilon _{\rm tot} (D+1)  \right) = (D+1) F^* \dot  D\,, 
\]
} 
which is indeed equivalent to the energy balance \eqref{energy_balance_1D} in the material frame from the relation \eqref{L_ell}, with $F=F^*/(D+1)$.
\rem{ 
This derivation justifies the form of energy conservation obtained earlier in \eqref{cons_energy} in the spatial frame, in which case $\epsilon _{\rm tot}$ reduces to just the internal energy $\epsilon(D, \rho  _s,s) = \rho_s e_s(D) + \epsilon_t(s)$.

For the Lagrangian \eqref{Lagrangian_1D} one gets, by using \eqref{L_ell},
\begin{equation}\label{ell_1D} 
\ell(D, \dot  D, \rho _k, s)= \frac{1}{2}   \frac{ M \rho\dot  D ^2 }{(D+1)^2}  - \epsilon (D, \rho  _k, s)\,,
\end{equation}
with $ \epsilon $ defined from $E$ as $E(D,\varrho _k,S)= \frac{1}{D+1} \epsilon \big(D, \varrho _k(D+1),S(D+1)\big)$.
} 

The quasi-static case studied in Section \ref{sec:1D_motion} corresponds to the Lagrangian 
\[
\ell( D, \dot  D, \rho  _k, s)= - \epsilon (D, \rho  _s, s)= - \frac{ \varrho _0}{D+1}  e _t\left(\frac{s(D+1)}{ \varrho _0} \right) - \rho  _s   e_{\rm el}(b) \,,
\]
see \eqref{final_ell},   for which system \eqref{eqs_spatial_discrete_2} recovers \eqref{rho_T_coupled_modified_new}. 
\rem{ 
Indeed, for the particular case when only the variables related to solid are entering the Lagrangian, and $q_{sb}=q$, system \eqref{eqs_spatial_discrete} reduces to 
\begin{equation} 
\left\{
\begin{array}{l}
\displaystyle\vspace{0.2cm} - \frac{\partial \epsilon }{\partial s} \bar D_t s = q\frac{\partial \epsilon }{\partial \rho  _s}\\
\displaystyle\vspace{0.2cm}  \bar D_t \rho  _s = q , \qquad \bar D_t \rho  _b = -q \,.
\\
\displaystyle\vspace{0.2cm}\frac{\partial \epsilon }{\partial D} (D+1) + \frac{\partial \epsilon }{\partial s} s +  \frac{\partial \epsilon }{\partial \rho  _s} \rho  _s - \epsilon  = F(D+1)\,.
\end{array}
\right.
\label{alt_momentum_eq} 
\end{equation} 
}








\color{black} 

\section{Conclusion} 

In this paper, we have developed a general, thermodynamically consitent, variational theory of porous media that are allowed to undergo breaking in the solid matrix under stress. Our theory extends the theory of continuous damage mechanics by providing a clear physical explanation for the damage parameter as the density of the broken material. In our theory, the damage is developed by the irreversible destruction of the elastic matrix component, while the microscopic properties of the remaining elastic component remain the same throughout the process. This approach removes the need to develop a phenomenological theory of dependence of elastic coefficients on the damage parameter.  We have also developed thermodynamically consistent expressions for the breaking rate of the solid matrix, along with the expressions for the friction forces and stress that are consistent with the second law of thermodynamics for general materials.  The form of the equations resulting
from the variational formulation was central in the  search of such  thermodynamically consistent expressions.
We have then considered simplified, but 
thermodynamically consistent, single velocity versions of our model, which allowed us to derive exact reduction to ODEs for the analysis of matrix break-up. Finally, we developed a general variational approach for the derivation of such thermodynamically consistent  multicomponent reduced models.

In a follow-up work, we intend to continue this process and develop particular analytical and numerical solutions for the equations \eqref{thermodynamics_lagr_particular} for different geometries, and different expressions of the elastic and thermal energies of the material. Of particular interest is the question of a sudden raise in fluid pressure in one point of the material and remaining propagation of the damage in the bulk, as was done in \cite{mobasher2017non}. Moreover, since our theory is variational, it is useful for deriving variational integrators for long-term simulations of damage dynamics \cite{bauer2017variational,eldred2018variational,brecht2019variational,GaGB2020,GaGB2023}.

\section*{Acknowledgements}
We are grateful for fruitful discussions to Profs. D. D. Holm and D. V. Zenkov, and Drs. C. Eldred and A. Lokhov, elucidating the applicability of our theory and the extent of development of  damage mechanics theories. VP was partially supported by NSERC Discovery grant 2023-03590.


\appendix

\section{Alternative expressions for heat and mass transfer coefficients involving cross-phenomena}\label{App_A}

We describe here possible cross-effects in the reversible processes, based on the form of the entropy production equation \eqref{total_entropy_particular}.

Regarding mass and heat transfer, cross-effects can occur for $k=s,\ell=b$, while still preserving \eqref{positivity_matter_heat}, which are the discrete analogues to the Soret and Dufour effects. We can thus consider 
\begin{equation} 
\label{positivity_matrix_general} 
\begin{bmatrix}
\vspace{0.1cm}J_{sb} \frac{T_s-T_b}{T_s T_b} \\
q_{sb}
\end{bmatrix}
=
\mathcal{L} _{sb}
\begin{bmatrix}
\vspace{0.1cm} T_b - T_s \\
\frac{1}{T_s}\big( \frac{1}{2} | \boldsymbol{u}_s| ^2 - g_s \big)- \frac{1}{T_b}\big( \frac{1}{2} | \boldsymbol{u}_b| ^2 - g_b\big)
\end{bmatrix}
\end{equation} 
with the symmetric part of the $2 \times 2$ matrix $\mathcal{L} _{sb}$ positive, and with $J_{sf}$ and $J_{bf}$ satisfying \eqref{heat_flux_cond} or, equivalently \eqref{heat_flux_cond_2}. 
We shall note that the conditions \eqref{general_q_sb} and \eqref{positivity_matrix_general} for the mass and heat transfer coefficients seem to be novel and not encountered in previous literature on the subject. On the other hand, since $q_{sb}$ can be interpreted physically as the rate of damage to the elastic matrix, these equations do have some aspects of the phenomenological equations for the evolution of damage parameters used previously in the literature, see \emph{e.g.} \cite{lyakhovsky2007damage}. The equations for transfer rate $q$ described there, in our notation, would be an affine function of the internal pressure $p$. Our expression has these pressure terms, but also involves additional terms related to the elastic energy, as we show below.

Besides the cross-phenomena between the scalar processes associated to mass and heat transfer, another cross-effect can be postulated, mathematically similar to the cross-phenomena between bulk viscosity and chemistry, see \cite{dGMa1969}.
Ignoring the friction forces and assuming that the stresses $ \boldsymbol{\sigma} _{k\ell}$ are symmetric, we can rewrite the entropy production expression \eqref{total_entropy_particular} as
\begin{align*} 
&\sum_{k<\ell}\Bigg[  \boldsymbol{\sigma} _{k\ell}^\circ: \Big( \frac{\mathbb{F} _k^\circ}{T_k}- \frac{\mathbb{F}  _\ell^\circ}{T_\ell} \Big) 
+  \frac{1}{3} \operatorname{Tr}( \boldsymbol{\sigma} _{k\ell})  \Big( \frac{\operatorname{div} \mathbf{u} _k}{T_k} - \frac{\operatorname{div} \mathbf{u} _\ell}{T_\ell} \Big)\\
& \qquad +  q_{k\ell}\Big(\frac{\frac{1}{2} | \boldsymbol{u}_k| ^2 - g_k}{T_k}-\frac{\frac{1}{2} | \boldsymbol{u}_\ell| ^2 - g_\ell}{T_\ell}  \Big) + J_{k\ell}\Big(\frac{1}{T_\ell} - \frac{1}{T_k}\Big)(T_\ell-T_k)\Bigg]\,.
\end{align*}
Therefore, one is naturally lead to postulate the following cross-effects of scalar processes:
\begin{equation} 
\label{cross_condition} 
\left[
\begin{array}{c}
\vspace{0.2cm} \frac{1}{3} \operatorname{Tr}( \boldsymbol{\sigma} _{k\ell})\\
\displaystyle\vspace{0.2cm} q_{k\ell}\\
J_{k\ell} \frac{T_k-T_\ell}{T_kT_\ell}
\end{array}
\right] 
= 
\mathcal{L}_{k\ell}
\left[
\begin{array}{c}
\vspace{0.2cm}\frac{1}{T_k}\operatorname{div} \mathbf{u} _k - \frac{1}{T_\ell}\operatorname{div} \mathbf{u} _\ell \\
\vspace{0.2cm}\frac{1}{T_k} (\frac{1}{2} | \boldsymbol{u}_k| ^2 - g_k)-\frac{1}{T_\ell}(\frac{1}{2} | \boldsymbol{u}_\ell| ^2 - g_\ell) \\
\displaystyle T_\ell-T_k
\end{array}
\right] \,,
\end{equation} 
where the symmetric part of the $3 \times 3$ matrices $ \mathcal{L}_{k\ell} $ must be positive from the second law. We shall not endeavor to consider the more general conditions \eqref{positivity_matrix_general} and \eqref{cross_condition} involving the cross effects, and only concentrate on the simpler conditions \eqref{q_cond_single} as the one leading to the most simple mathematical expressions treatable analytically. In spite of the relative mathematical simplicity compared to \eqref{positivity_matrix_general} and \eqref{cross_condition}, expressions \eqref{q_cond_single}  lead to physically relevant systems providing physically valid quantitative predictions.

\end{document}